\documentclass[11pt]{article}
\usepackage[usenames,dvipsnames,svgnames,table]{xcolor} 

\usepackage{enumitem} 
\usepackage{rotfloat} 
\usepackage{graphicx}
\usepackage{epsfig}
\usepackage{amssymb, amsmath}
\usepackage{verbatim}
\usepackage{natbib}
\usepackage{authblk}
\usepackage{kotex}
\usepackage{multirow}
\usepackage[colorlinks]{hyperref} 
\usepackage[colorinlistoftodos, textsize=scriptsize]{todonotes} 
\usepackage{subfigure} 

\newtheorem{theorem}{Theorem}[section]

\newenvironment{proof}[1][Proof]{\begin{trivlist}
		\item[\hskip \labelsep {\bfseries #1}]}{\end{trivlist}}

\newenvironment{remark}[1][Remark]{\begin{trivlist}
		\item[\hskip \labelsep {\bfseries #1}]}{\end{trivlist}}

\newcommand{\bch}{\color{black}  }   
\newcommand{\ech}{\color{black}  ~}    

\DeclareMathOperator*{\argmax}{argmax}

\newcommand{\bea}{\begin{eqnarray*}}
	\newcommand{\eea}{\end{eqnarray*}}
\newcommand{\bean}{\begin{eqnarray}}
\newcommand{\eean}{\end{eqnarray}}

\newcommand{\bfX}{{\bf X}}

\newcommand{\V}{{\rm Var}}
\newcommand{\C}{{\rm Cov}}

\newcommand{\sg}{\Sigma}
\newcommand{\what}{\widehat}

\newcommand{\lra}{\longrightarrow}

\newcommand{\bbP}{\mathbb{P}} 
\newcommand{\bbR}{\mathbb{R}}
\newcommand{\bbE}{\mathbb{E}}

\begin{document}
	
	\title{Scalable Bayesian high-dimensional local dependence learning}
	
	\author[1]{Kyoungjae Lee}
	\affil[1]{Department of Statistics, Sungkyunkwan university}
	\author[2]{Lizhen Lin}
	\affil[2]{Department of Applied and Computational Mathematics and Statistics, University of Notre Dame}
	
	\maketitle

	\begin{abstract}
		In this work, we propose a scalable Bayesian procedure for learning the local dependence structure in a high-dimensional model where the variables possess a natural ordering. The ordering of variables can be indexed by time, the vicinities of spatial locations, and so on, with the natural assumption that variables far apart tend to have weak correlations. Applications of such models abound in a variety of fields such as finance, genome associations analysis and spatial modeling. We adopt a flexible framework under which each variable is dependent on its neighbors or predecessors, and the neighborhood size can vary for each variable. It is of great interest to reveal this local dependence structure by estimating the covariance or precision matrix while yielding a consistent estimate of the varying neighborhood size for each variable. The existing literature on banded covariance matrix estimation, which assumes a fixed bandwidth cannot be adapted for this general setup. We employ the modified Cholesky decomposition for the precision matrix and design a flexible prior for this model through appropriate priors on the neighborhood sizes and Cholesky factors. The posterior contraction rates of the Cholesky factor are derived which are nearly or exactly minimax optimal, and our procedure leads to consistent estimates of the neighborhood size for all the variables. Another appealing feature of our procedure is its scalability to models with large numbers of variables due to efficient posterior inference without resorting to MCMC algorithms. Numerical comparisons are carried out with competitive methods, and applications are considered for some real datasets.
		

	\end{abstract}
	

	\section{Introduction}
	
	The problem of covariance matrix or precision matrix estimation has been extensively studied over the last few decades.  A typical model setup is to assume $X=(X_1,\ldots, X_p)^T \in \bbR^p$  follows  a $p$-dimensional normal distribution   $N_p(0, \sg)$, with mean zero and covariance matrix $\sg \in\bbR^{p\times p}$. The dependence structure among the variables is encoded by  the covariance  matrix $\sg$ or its inverse $\Omega = \sg^{-1}$. 
	In high-dimensional settings where $p$ can be much larger than the sample size, the traditional sample covariance matrix or inverse-Wishart prior leads to inconsistent estimates of $\Sigma$ or $\Omega$, see \cite{johnstone2009consistency} and \cite{lee2018optimal}. Restricted matrix classes   with a banded or sparse structure are often imposed on the covariance  or precision matrices   as a common practice for consistent estimation (see, e.g., \cite{bickel2008regularized}, \cite{cai2016estimating} and \cite{lee2019minimax}).


	In this paper, we focus on investigating the local dependence structure in a high-dimension model where the variables possess a natural ordering.  An ordering on variables is often encountered for example in time series or genome data, where variables close to each other in time or location are more likely to be correlated than variables located far apart.  More specifically, we assume that \bch each variable depends on its neighboring variables or predecessors\ech and more importantly,  the \emph{the size of the neighborhood  or bandwidth can vary with each variable.}  Therefore, the existing literature on banded covariance matrix estimation which deals with a fixed bandwidth cannot be adapted for this set up. 
	
	\bch Our work employs the modified Cholesky decomposition (MCD) \citep{pourahmadi1999joint} of the precision matrix, which provides an efficient way to learn the local dependence  structure of the data. 
	The MCD has been widely used for covariance or precision matrix estimation including  in \cite{rutimann2009high}, \cite{shojaie2010penalized} and \cite{van2013ell}, just to name a few.\ech
	In more details, assume \bch the variables $X_1,\ldots, X_p$ are\ech arranged according to a known ordering.
	For any $p\times p$ positive definite matrix $\Omega$, \bch there uniquely exist lower triangular matrix $A = (a_{jl}) \in \bbR^{p\times p}$ and diagonal matrix $D= diag(d_j) \in \bbR^{p\times p}$\ech such that  $\Omega  = (I_p - A)^T D^{-1}  (I_p-  A)$, where $a_{jj} = 0$ and $d_j>0$ for all $j=1,\ldots, p$.
	It is called the MCD of $\Omega$, and we call $A$ the Cholesky factor. 
	Based on the MCD, $X = ( X_1,\ldots, X_p)^T \sim N_p(0, \Omega^{-1})$ is equivalent to a set of linear regression models, $X_1 \sim N(0, d_1)$ and 
	\bean
	X_j \mid X_1,\ldots, X_{j-1} &\sim& N \Big(  \sum_{l=1}^{j-1} a_{jl} X_l   , \,\,  d_j  \Big), \quad j=2,\ldots ,p  . \label{linear_reg}
	\eean
	In this paper, we assume \emph{local dependence structure} by considering 
	\bean\label{linear_reg_local}
	X_j \mid X_{j- k_j},\ldots, X_{j-1} &\sim& N \Big(  \sum_{l=j - k_j}^{j-1} a_{jl} X_l   ,  \,\,  d_j \Big), \quad j=2,\ldots ,p  
	\eean
	for some $k_j$ instead of \eqref{linear_reg}, which implies $a_{jl} =0 $ for any $l = 1, \ldots, j-k_j-1$ and $j=2,\ldots, p$.
	It leads to a varying bandwidth structure for the Cholesky factor which means that $X_j$ is dependent only on $k_j$ closest variables, $ X_{j- k_j},\ldots, X_{j-1}$, and conditionally independent of others.
	Our goal is to learn the local dependence structure of data based on model \eqref{linear_reg_local} by learning $a_{jl}, d_j$ as well as the varying bandwidth or neighborhood size $k_j$. \bch  Allowing a varying bandwidth in the above model is a more realistic assumption as the range and pattern of dependence can vary over time or spatial locations.\ech
	See Figure \ref{fig:A0s} for an illustration of the Cholesky structure with varying bandwidth.

	A number of work have been proposed for estimating sparse Cholesky factors based on penalized likelihood approaches, including \cite{rothman2010new}, \cite{van2013ell} and \cite{khare2019scalable}.
	However, these methods are not suitable for modeling local dependence  because they allow arbitrary sparsity patterns for Cholesky factors.
	For example, $X_j$ can depend on $X_1$ but not on $X_2,\ldots, X_{j-1}$  in their framework, which is not desirable in the context of local dependence.
	\cite{an2014hypothesis} considered a banded Cholesky factor and proposed a consistent test for bandwidth selection, but they assumed the same local dependence structure for all variables, that is, assumed a common $k$ instead of $k_j$ in \eqref{linear_reg_local}.	
	The most relevant work is \cite{yu2017learning}  which developed a penalized likelihood approach for estimating the banded Cholesky factor in \eqref{linear_reg_local} using a hierarchical group lasso penalty.
	Their approach is formulated as a convex optimization problem, and theoretical properties such as selection consistency and convergence rates were established.
	However, they required relatively strong conditions to obtain theoretical results, which will be discussed in detail in Section \ref{sec_main}.

	From a Bayesian perspective, \cite{banerjee2014posterior} and \cite{lee2017estimating} suggested $G$-Wishart priors and banded Cholesky priors for learning local dependence structure with a common bandwidth $k$, and posterior convergence rates for precision matrices are derived assuming $k$ is known. 
	\cite{lee2019minimax} proposed the empirical sparse Cholesky prior for sparse Cholesky factors.
	Under their model each variable can be dependent on distant variables, so it is not suitable for local dependence structure.
	\cite{lee2018bayesian} suggested a prior distribution for banded Cholesky factors.
	Although  bandwidth selection consistency  as well as consistency of Bayes factors were established in high-dimensional settings,  again a common bandwidth $k$ is assumed in their model. Achieving selection consistency simultaneously  for  bandwidths of all the variables is a much more challenging problem as the number of bandwidths as well as the size of the bandwidth can diverge as the number of variables $p$ goes to infinity.


	In this paper, we propose a prior tailored to Cholesky factors with local dependence structure. 
	The proposed prior allows a flexible learning of local dependence structure based on varying bandwidth $k_j$.
	This paper contributes in both theoretical and practical developments.
	For theoretical advancement, we prove selection consistency for varying bandwidth and \bch nearly\ech optimal posterior convergence rates for Cholesky factors.
	This is the first Bayesian method for local dependence learning with varying bandwidth in high-dimensional settings with theoretical guarantees. 
	Furthermore, we significantly weaken required conditions to obtain theoretical properties compared with  those in \cite{yu2017learning}.
	\bch On the other hand, from a practical point of view, the induced posterior allows fast computations, enabling scalable inference for large data sets.\ech	
	The posterior inference does not require Markov chain Monte Carlo (MCMC) algorithms and is easily parallelizable.	
	Furthermore, our simulation studies show that the proposed method outperforms other competitors in various settings.
	We find that the proposed cross-validation for the proposed method is much faster than the contenders and selects nearly optimal hyperparameter in terms of specificity and sensitivity.
	Finally, it is worth mentioning that posterior inference for sparse Cholesky factors with an arbitrary sparsity pattern  is computationally much more expensive than that for banded Cholesky factors, developing a statistical method for banded Cholesky factors is therefore of great importance independent of existing methods for sparse Cholesky factors.

	

	The rest of paper is organized as follows.
	In Section \ref{sec:prel}, model assumptions, the proposed local dependence Cholesky prior and the induced fractional posterior are introduced.
	The main results including bandwidth selection consistency and optimal posterior convergence rates are established in Section \ref{sec_main}.
	In Section \ref{sec_simul}, the performance of the proposed method is illustrated based on simulated data and real data analysis.
	\bch Concluding remarks and discussions are given in Section \ref{sec_disc}, while the proofs of main results and additional simulation results are provided in the Appendix.
	\verb|R| codes for implementation of our empirical results are available at \url{https://github.com/leekjstat/LANCE}.\ech

	\section{Preliminaries}\label{sec:prel}
	
	\subsection{Notation}
	For any $a$ and $b\in\bbR$, we denote $a \wedge b$ and $a\vee b$ as the minimum and maximum of $a$ and $b$, respectively.
	For any positive sequences $a_n$ and $b_n$, $a_n = o(b_n)$ denotes $a_n / b_n \lra 0$ as $n\to\infty$.
	We denote $a_n = O(b_n)$, or equivalently $a_n \lesssim b_n$, if there exists a constant $C>0$ such that $a_n < C b_n$ for all large $n$.	
	For any matrix $A = (a_{ij})\in \bbR^{p\times p}$, we denote $\lambda_{\min}(A)$ and $\lambda_{\max}(A)$ as the minimum and maximum eigenvalues, respectively. 
	\bch Furthermore, we define the matrix $\ell_\infty$-norm, Frobenius norm and element-wise maximum norm as follows:
	$\|A\|_\infty = \max_{1\le i \le p} \sum_{j=1}^p |a_{ij}|$, $\|A\|_F = ( \sum_{i=1}^p \sum_{j=1}^p a_{ij}^2  )^{1/2}$ and $\|A\|_{\max} = \max_{1\le i,j\le p} |a_{ij}|$.\ech
	We denote $IG(a,b)$ as the inverse-Gamma distribution with shape parameter $a>0$ and scale parameter $b>0$.

	\subsection{High-dimensional Gaussian Models}\label{subsec:model}
	Throughout the paper, we assume a high-dimensional setting with $p = p_n \ge n$, and that the variables have a known natural ordering.
	Specifically,  we  assume we observe   a sample of data  with sample size $n$ from a $p$-dimensional Gaussian model,
	\bean\label{model}
	X_1,\ldots, X_n \mid \Omega_n &\overset{i.i.d.}{\sim}&  N_p (0 ,  \Omega_n^{-1}),
	\eean
	where $\Omega_n = \sg_n^{-1} \in \bbR^{p\times p}$ is a precision matrix.
	For the rest of the paper, we use subscript $n$ for any $p\times p$ matrices to indicate that the dimension $p=p_n$ grows as $n\to\infty$.
	Let $\bfX_n = (X_1,\ldots, X_n)^T \in \bbR^{n\times p}$ be a data matrix, and $X_i =(X_{i1}, \ldots, X_{ip})^T \in \bbR^p$ for all $i=1,\ldots, n$.
	We denote the Cholesky factor and the diagonal matrix from the MCD as $A_n$ and $D_n$, respectively, i.e., $\Omega_n = (I_p - A_n )^T D_n^{-1}(I_p - A_n)$.
	Model \eqref{model} is related to a directed acyclic graph (DAG) model depending on the  sparsity pattern of $A_n$ \citep{van2013ell, lee2019minimax}, but we will not go into detail on DAG models.

	In this paper, we assume $A_n = (a_{jl})$ has a banded structure with {\it varying bandwidths}, $\{k_2,\ldots, k_p\}$, which satisfies $0\le k_j \le j-1$ and $\sum_{l : |j-l| > k_j} |a_{jl}|=0$ for each $j=2,\ldots,p$.
	Let $\tilde{X}_j \in \bbR^n$ and $\bfX_{j (k_j)} \in \bbR^{n \times k_j}$ be sub-matrices of $\bfX_n$ consisting of the $j$th and $(j- k_j),\ldots, (j-1)$th columns, respectively.
	Under the varying bandwidths assumption, model \eqref{model} can be represented as 
	\bea
	\tilde{X}_1 \mid d_1 &\sim& N_n( 0 , d_1 I_n) , \\
	\tilde{X}_j \mid \bfX_{j (k_j)} , a_j^{(k_j)} , d_j, k_j &\sim& N_n \big(  \bfX_{j (k_j)} a_j^{(k_j)} ,    d_j I_n  \big), \quad j=2,\ldots, p,
	\eea
	where $a_j^{(k_j)} = ( a_{jl} )_{(j- k_j) \le l \le (j-1)} \in \bbR^{k_j}$.
	The above representation implies that to predict the $j$th variable, $\tilde{X}_j$, it suffices to know its $k_j$-nearest predecessors, $\bfX_{j (k_j)}$.
	Thus, the varying bandwidths assumption of $A_n$ directly induces the {\it local dependence} structure between variables.
	As mentioned before, this is often natural given commonly encountered ordered variables in time series or genome data sets, and \bch a more realistic and flexible assumption than the common bandwidth assumption.\ech

	\subsection{Local Dependence Cholesky Prior}\label{subsec:prior}	
	To conduct Bayesian inference, we need to impose a prior on $A_n$ and $D_n$ as well as $k_j$, $j=2,\ldots, p$, which restricts $A_n$ to have a local dependence structure.
	We propose the following prior 
	\bean\label{prior}
	\begin{split}
		a_j^{(k_j)} \mid d_j, k_j \,\,&\overset{ind.}{\sim}\,\,   N_{k_j} \Big(  \what{a}_j^{(k_j)} ,    \frac{d_j}{\gamma}\big( \bfX_{j (k_j)}^T \bfX_{j (k_j)} \big)^{-1}     \Big) , \,\, j=2,\ldots, p,   \\
		\pi(d_j)  \,\,&\propto\,\,  d_j^{-\nu_0/2 - 1}, \,\, j=1,\ldots,p ,  \\
		\pi(k_j) \,\,&\propto\,\,  c_1^{- k_j} p^{-c_2 k_j}  I(0 \le k_j \le  \{R_j \wedge(j-1)\}  ) , \,\, j=2,\ldots, p ,
	\end{split}
	\eean
	for some positive constants $\gamma , c_1,c_2, R_2,\ldots, R_p$ and $\nu_0$, where $\what{a}_j^{(k_j)}  = (\bfX_{j (k_j)}^T \bfX_{j (k_j)})^{-1}  \bfX_{j (k_j)}^T \tilde{X}_j$.
	We call the above prior LANCE (LocAl depeNdence CholEsky) prior.
	The conditional prior for $a_j^{(k_j)}$ is a version of the Zellner's g-prior \citep{zellner1986assessing} and depends on the data.
	By using the prior $\pi(a_{j}^{(k_j)} \mid d_j, k_j)$ with sub-Gaussian tails and data-dependent center, we can obtain theoretical properties of posteriors without having to use priors with heavier tails or introducing redundant upper bound conditions on $\|a_j^{(k_j)} \|_2$ as discussed by \cite{lee2019minimax}.
	The prior for $d_j$ is improper and includes the Jeffreys' prior \citep{jeffreys1946invariant}, $\pi(d_j) \propto d_j^{-1}$, as a special case.
	In fact, the proper prior $d_j \sim IG(\nu_0/2 , \nu_0')$, for some constant $\nu_0'>0$, can be used.
	However, we proceed with the above improper prior to reduce the number of hyperparameters.

	For posterior inference, we suggest using the fractional posterior, which has received increased attention recently \citep{martin2014asymptotically,martin2017empirical,lee2019minimax}.
	Let $L(A_n, D_n)$ be the likelihood function of model \eqref{model}.	
	For a given constant $0<\alpha<1$, the $\alpha$-fractional posterior is defined by 	$\pi_\alpha( A_n , D_n \mid \bfX_n)  \propto L(A_n, D_n)^\alpha \pi (A_n, D_n)$.
	Thus, $\alpha$-fractional posterior is a posterior distribution updated with the likelihood function raised to a power of $\alpha$ instead of the usual likelihood.
	We denote the $\alpha$-fractional posterior by $\pi_\alpha(\cdot \mid \bfX_n)$ to indicate the use of $\alpha$-fractional likelihood.
	Theoretical properties of fractional posterior can  often  be established under weaker conditions compared with the usual posterior \citep{bhattacharya2019bayesian}.
	Under our model, the $\alpha$-fractional posterior has the following closed form:
	\bean\label{posteriors}
	\begin{split}
		a_{j}^{(k_j)} \mid d_j, k_j, \bfX_n \,\,&\overset{ind.}{\sim}\,\,  N_{k_j} \Big(  \what{a}_j^{(k_j)} ,    \frac{d_j}{\alpha+\gamma}\big( \bfX_{j (k_j)}^T \bfX_{j (k_j)} \big)^{-1}     \Big) , \,\, j=2,\ldots, p,   \\
		d_j \mid k_j , \bfX_n \,\,&\overset{ind.}{\sim}\,\, IG \Big(  \frac{\alpha n + \nu_0}{2} , \, \frac{\alpha n}{2}\what{d}_{j}^{(k_j)}  \Big) , \,\, j=1,\ldots, p , \\
		\pi_\alpha(k_j \mid \bfX_n)  \,\,&\propto\,\,  \pi(k_j)  \Big( 1+ \frac{\alpha}{\gamma}  \Big)^{- \frac{k_j}{2}}  \big(  \what{d}_{j}^{(k_j)} \big)^{- \frac{\alpha n + \nu_0}{2}}  , \,\, j=2,\ldots,p  ,
	\end{split}
	\eean
	where $\what{d}_{j}^{(k_j)} = n^{-1} \tilde{X}_j^T ( I_n - \tilde{P}_{j k_j}) \tilde{X}_j$ and $\tilde{P}_{j k_j} = \bfX_{j (k_j)} (\bfX_{j (k_j)}^T \bfX_{j (k_j)})^{-1} \bfX_{j (k_j)}^T$.
	Based on \eqref{posteriors}, one can notice that posterior inference for each $j=2,\ldots, p$ is parallelizable.
	Furthermore, a MCMC algorithm is not needed because direct posterior sampling from \eqref{posteriors} is possible.
	Note that the three posterior distributions in \eqref{posteriors} form the joint posterior distributions of $(a_j^{(k_j)} , d_j, k_j)$, which are independent for each $j=1,\ldots, p$.
	Thus,  LANCE prior leads to a fast and scalable posterior inference even in high-dimensions.

	\section{Main Results}\label{sec_main}
	
	In this section, we show that LANCE prior accurately unravels the local dependence structure.
	More specifically, it is proved that LANCE prior attains bandwidth selection consistency  for all bandwidths and nearly minimax posterior convergence rates for Cholesky factors.	
	To obtain desired asymptotic properties of posteriors, we assume the following conditions:
	\begin{itemize}
		\item[(A1)] $\lambda_{\max}(\Omega_{0n}) \log p /\lambda_{\min}(\Omega_{0n})  = o(n)$.
		
		\item[(A2)] For some constant $M_{\rm bm} > c_2 +1$, 
		\bea
		\min_{(j,l): a_{0,jl} \neq 0 } \frac{a_{0,jl}^2}{d_{0j}}  &\ge& \frac{10 M_{\rm bm} \,\, \lambda_{\max}(\Omega_{0n}) }{(\alpha + \nu_0/n )(1- \alpha - \nu_0/n) }  \frac{\log p}{n} .
		\eea
		
		\item[(A3)]  $k_{0j} \le R_j$ for any $j=2,\ldots, p$,  $\max_j R_j  \log p \le n (1+ 5\sqrt{\epsilon}) / \{C_{\rm bm}(1- 2 \epsilon)^2\} $, where $\epsilon = \{(1-\alpha)/10 \}^2$ and $C_{\rm bm} =  10 M_{\rm bm} / \{(\alpha + \nu_0/n )(1- \alpha - \nu_0/n)\} $

		\item[(A4)] The hyperparameters satisfy $\nu_0 = O(1)$, $\gamma =O(1)$, $c_1=O(1)$, $c_2>1$ and $0.6 \le \alpha <1$.
	\end{itemize}
	
	Condition (A1) allows the condition number $\lambda_{\max}(\Omega_{0n})/\lambda_{\min}(\Omega_{0n})$ to grow to infinity at a rate slower than $n/ \log p$.
	\bch This condition is weaker than the bounded eigenvalue conditions\ech used in \cite{banerjee2015bayesian}, \cite{khare2019scalable} and \cite{lee2019minimax}, which assume $c< \lambda_{\min}(\Omega_{0n}) \le \lambda_{\max}(\Omega_{0n}) < C$ for some constants $c$ and $C>0$.
	\bch We note here that we require the bounded eigenvalue conditions for the (nearly) minimaxity of the posterior convergence rates in Theorems \ref{thm_post_conv} and \ref{thm_post_conv_nobetamin}.\ech

	Condition (A2) determines the lower bound for the nonzero signals, $a_{0,jl}^2/ d_{0j}$.
	This is called the {\it beta-min} condition.
	It has been well known that the beta-min condition is essential for variable selection \citep{buhlmann2011statistics,martin2017empirical} and support recovery for sparse matrices \citep{yu2017learning,cao2019posterior}.
	Condition (A2) implies that the rate of nonzero $a_{0,jl}^2/ d_{0j}$ should, at least, be $\lambda_{\max}(\Omega_{0n}) \log p /n$, which has the same rate with $\log p/n$ under the bounded eigenvalue condition.
	
	Condition (A3), together with condition (A4), provides an upper bound for the true bandwidth $k_{0j}$: they allow the maximum bandwidth, $\max_j k_{0j}$, to grow to infinity at a rate not faster than $n/\log p$.
	The rest of condition (A4) presents sufficient conditions for hyperparameters to obtain theoretical properties.
	The conditions on $c_1$ and $c_2$ control the strength of penalty for large models, i.e., large bandwidths.
	Note that condition $c_2>1$ is weaker than the condition $c_2 \ge 2$ used in \cite{lee2019minimax}, which implies that the local dependence assumption requires a weaker penalty compared with arbitrary sparsity patterns.
	In Section \ref{sec_simul}, we will give a practical guidance for the choice of hyperparameters.

	\begin{theorem}[Bandwidth selection consistency]\label{thm_bandwidth_sel}
		Consider model \eqref{model} and LANCE prior \eqref{prior}.
		Under conditions (A1)--(A4), we have
		\bea
		\bbE_0 \Big\{  \pi_\alpha \big(  k_j \neq k_{0j} \text{ for at least one } 2\le j \le p  \mid \bfX_n \big)  \Big\} &\lra&  0  \quad \text{ as } n\to\infty.
		\eea
	\end{theorem}
	
	Theorem \ref{thm_bandwidth_sel} says that the posterior probability of incorrectly estimating local dependence structure, i.e., bandwidths,  converges to zero in probability as $n\to\infty$.
	Thus, the proposed method can consistently recover local dependence structure for each variable.

	We compare the above result with existing theoretical results in other work.
	First of all, we note here that \cite{khare2019scalable}, \cite{cao2019posterior} and \cite{lee2019minimax} assumed arbitrary dependence structures for Cholesky factors, thus their methods are not tailored to local dependence structure considered in this paper.
	Although \cite{bickel2008regularized}, \cite{banerjee2014posterior} and \cite{lee2017estimating} focused on banded Cholesky factors, which result in banded precision matrices, they assumed a common dependence structure, i.e., a common bandwidth, for each variable, and did not provide bandwidth selection consistency.
	
	To the best of our knowledge, the state-of-the-art theoretical result for estimating local dependence structure is obtained by \cite{yu2017learning}.
	They proposed a penalized likelihood approach and obtained bandwidth selection consistency in a high-dimensional setting.
	They assumed that the eigenvalues of $\Omega_{0n}$ lie in $[\kappa^{-2}, \kappa^2]$ and the beta-min condition, 
	\bea
	\min_{(j,l): a_{0,jl} \neq 0 } \frac{|a_{0,jl}|}{\sqrt{d_{0j}} } &\ge& 8 \rho^{-1} ( 4 \max_j \| \sg_{0n, k_{0j}}^{-1} \|_\infty + 5 \kappa^2 )  \sqrt{2 \|D_{0n}\|_\infty \frac{\log p}{n} } ,
	\eea
	for some constants $\kappa >1$ and $\rho \in (0,1]$, where $\sg_{0n, k_{0j}} = ( \sigma_{0, il} )_{j -k_{0j} \le i, l \le j}$ denotes the sub-matrix of the true covariance matrix.
	Note that these conditions are more restrictive than our conditions (A1) and (A2).
	For example,  $\| \sg_{0n, k_{0j}}^{-1} \|_\infty = O( k_{0j}^{1/2} )$  holds under the bounded eigenvalue condition, so the beta-min condition in \cite{yu2017learning} implies that the minimum nonzero $a_{0,jl}^2/d_{0j}$ is bounded below by $\max_j k_{0j} \log p /n$ with respect to a constant multiple; in contrast, the rate of the lower bound in condition (A2) is $\log p /n$ under the bounded eigenvalue condition.
	Furthermore, they assumed the so-called {\it irrepresentable} condition, 
	$$ \max_{2\le j\le p} \max_{1\le l \le j -k_{0j}-1 }  \| ( \sg_{0n} )_{l, k_{0j}} \sg_{0n, k_{0j}}^{-1} \|_1  \le  \frac{6(1-\rho)}{\pi^2} ,$$
	which is typically required for the lasso type methods with a random design matrix (e.g., see \cite{wainwright2009sharp} and \cite{khare2019scalable}).
	\cite{yu2017learning} proved the exact signed support recovery property under the above conditions and $n > \rho^{-2} \|D_{0n}\|_\infty \kappa^2$ $(12\pi^2 \max_j k_{0j} +32 ) \log p$.
	Note that condition (A3) together with (A4) implies $n > C\max_j k_{0j} \log p$ for some constant $C>0$.
	Hence,  stronger conditions  are also required  by \cite{yu2017learning}  to establish bandwidth selection consistency compared with Theorem \ref{thm_bandwidth_sel}.

	Next, we show that LANCE prior achieves nearly minimax posterior convergence rates for Cholesky factors.
	The posterior convergence rates are obtained with or without beta-min condition (A2).
	\bch Under the beta-min condition, Theorem \ref{thm_post_conv} presents posterior convergence rates based on various matrix norms.\ech
	
	\begin{theorem}[Posterior convergence rates with beta-min condition]\label{thm_post_conv}
		Suppose that the conditions in Theorem \ref{thm_bandwidth_sel} hold.
		\bch If $k_0 \log p = o(n)$, where $k_0 = \max_j k_{0j}$,  we have
		\bea
		\bbE_0 \Big[ \pi_\alpha \Big\{   \|A_n - A_{0n}\|_{\max} \ge K_{\rm chol}  \frac{\lambda_{\max}(\Omega_{0n})^2 }{\lambda_{\min}(\Omega_{0n})^2 }  \Big(  \frac{k_0+ \log p }{n} \Big)^{1/2}   \mid\bfX_n  \Big\}  \Big]  &=&  o(1)  , \\
		\bbE_0 \Big[ \pi_\alpha \Big\{   \|A_n - A_{0n}\|_\infty \ge K_{\rm chol} \frac{\lambda_{\max}(\Omega_{0n})^2 }{\lambda_{\min}(\Omega_{0n})^2 }  \sqrt{k_0} \Big(  \frac{k_0+ \log p }{n} \Big)^{1/2}   \mid\bfX_n  \Big\}  \Big]  &=&  o(1)  , \\
		\bbE_0 \Big[ \pi_\alpha \Big\{   \|A_n - A_{0n}\|_F^2 \ge K_{\rm chol} \frac{\lambda_{\max}(\Omega_{0n})^2 }{\lambda_{\min}(\Omega_{0n})^2 }  \frac{  \sum_{j=2}^p (k_{0j}+ \log j )  }{n}     \mid\bfX_n  \Big\}  \Big]  &=&  o(1) ,
		\eea
		for some constant $K_{\rm chol}>0$ not depending on unknown parameters.\ech
	\end{theorem}

	Estimating each row of Cholesky factor $A_n$ can be considered as a linear regression problem with a random design matrix.
	By assuming beta-min condition (A2), we can use the selection consistency result in Theorem \ref{thm_bandwidth_sel}, although a beta-min condition is usually not essential for obtaining convergence rates.
	With a cost of the beta-min condition, we can focus only on the set $a_j^{(k_j)} =a_j^{(k_{0j})}$ for all $j$ in the posterior with high probability tending to 1.
	Then the above posterior convergence rates for various matrix norms boil down to those related to $\| a_j^{(k_{0j})} - a_{0j}^{(k_{0j})}\|$ for various vector norms $\|\cdot\|$, which makes the problem simpler.

	\bch 
	Assume that $\epsilon_0\le \lambda_{\min}(\Omega_{0n})\le \lambda_{\max}(\Omega_{0n})\le \epsilon_0^{-1}$ for some small constant $0<\epsilon_0 <1/2$.
	Then, the obtained posterior convergence rates in Theorem \ref{thm_post_conv} under the matrix $\ell_\infty$-norm and Frobenius norm are minimax if $\log p  = O(k_0)$ and $\log j = O(k_{0j})$ for all $j=2,\ldots, p$, respectively, by Theorem 3.3 in \cite{lee2019minimax}.
	Therefore, the above posterior convergence rates are nearly or exactly minimax depending on the dimensionality $p$ and bandwidths $k_{0j}$ under bounded eigenvalue conditions on $\Omega_{0n}$.
	\ech

	Now, we establish posterior convergence rate of Cholesky factors without beta-min condition (A2).
	To obtain the desired result, we consider a {\it modified} LANCE prior using $d_j \sim IG(\nu_0/2 , \nu_0')$ for some constant $\nu_0'>0$ instead of $\pi(d_j) \propto d_j^{-\nu_0/2-1}$ in \eqref{prior}.
	This modification is mainly used to derive a lower bound of the likelihood ratio appearing in the denominator of posteriors (see Lemma 7.1 in \cite{lee2019minimax}).
	Theorem \ref{thm_post_conv_nobetamin} shows the posterior convergence rate under various matrix norms.

	\begin{theorem}[Posterior convergence rates without beta-min condition]\label{thm_post_conv_nobetamin}
		Consider model \eqref{model} and the modified LANCE prior described above.
		Suppose that the conditions (A1), (A3) and (A4) hold.
		\bch If $ k_{0}\log p = o(n)$, $\lambda_{\max}(\Omega_{0n})/ \lambda_{\min}(\Omega_{0n})= O(p)$ and $\lambda_{\max}(\Omega_{0n})= o(n)$, we have
		\bea
		\bbE_0 \Big[ \pi_\alpha \Big\{   \|A_n - A_{0n}\|_{\max} \ge K_{\rm chol}  \Big( \frac{ \lambda_{\max}(\Omega_{0n}) }{\lambda_{\min}(\Omega_{0n})}  \frac{k_0\log p }{n} \Big)^{1/2}   \mid\bfX_n  \Big\}  \Big]   &=&  o(1)  , \\
		\bbE_0 \Big[ \pi_\alpha \Big\{   \|A_n - A_{0n}\|_\infty \ge K_{\rm chol} {k_0} \Big( \frac{ \lambda_{\max}(\Omega_{0n}) }{\lambda_{\min}(\Omega_{0n})}  \frac{\log p }{n} \Big)^{1/2}   \mid\bfX_n  \Big\}  \Big]  &=&  o(1)  ,  \\
		\bbE_0 \Big[ \pi_\alpha \Big\{   \|A_n - A_{0n}\|_F^2 \ge K_{\rm chol}  \frac{ \lambda_{\max}(\Omega_{0n}) }{\lambda_{\min}(\Omega_{0n})} \frac{  \sum_{j=2}^p k_{0}\log j   }{n}     \mid\bfX_n  \Big\}  \Big]   &=&  o(1) , 
		\eea
		for some constant $K_{\rm chol}>0$ not depending on unknown parameters.\ech
	\end{theorem}

	\bch 
	If we assume bounded eigenvalue conditions on $\Omega_{0n}$, the above posterior convergence rates are slightly slower than those in Theorem \ref{thm_post_conv} due to the absence of the beta-min condition.	
	Suppose $\epsilon_0\le \lambda_{\min}(\Omega_{0n})\le \lambda_{\max}(\Omega_{0n})\le \epsilon_0^{-1}$ for some small constant $0<\epsilon_0 <1/2$.
	Under these bounded eigenvalue conditions, \cite{yu2017learning} (Lemma 17) obtained the convergence rates $\zeta_\Gamma \sqrt{\log p /n}$, $\zeta_\Gamma (k_0 +1) (\log p /n)^{1/2}$ and $\zeta_\Gamma \sqrt{ ( \sum_{j=2}^p k_{0j} + p ) \log p /n}$ under the element-wise maximum norm, matrix $\ell_\infty$-norm and Frobenius norm, respectively, where $\zeta_\Gamma = 8 (2 \|D_{0n}\|_\infty )^{1/2} \rho^{-1} (  4 \max_j  \| \sg_{0n, k_{0j}}^{-1}\|_\infty   + 5 \epsilon_0^{-1} )$ for some constant $\rho \in (0,1]$.
	Note that, as mentioned before, it holds that $ \| \sg_{0n, k_{0j}}^{-1}\|_\infty  = O(k_{0j}^{1/2})$ without further assumption.
	Thus, their convergence rates are slower than or comparable to ours.
	We would also like to mention that, under bounded eigenvalue conditions on $\Omega_{0n}$, the posterior convergence rate under the matrix $\ell_\infty$-norm in Theorem \ref{thm_post_conv_nobetamin} is minimax if $\log p \asymp \log(p / k_0)$ by Theorem 3.5 in \cite{lee2019minimax}.\ech
	For example, it is the case if $k_0 = O(p^\beta)$ for some $0<\beta<1$.

	\begin{remark}
		By carefully modifying the proof of Theorem 3.6 in \cite{lee2019minimax}, we can probably obtain posterior convergence rates for precision matrices.
		However, it is not clear whether the obtained posterior convergence rates are minimax optimal.
		Although \cite{liu2020minimax} showed minimax rates for precision matrices with bandable Cholesky factors, they considered slightly different parameter spaces for Cholesky factors from what we consider.
		To the best of our knowledge, minimax rates for precision matrices based on Cholesky factors with varying bandwidths have not been established.
	\end{remark}

	\section{Numerical Studies}\label{sec_simul}

	\subsection{Choice of hyperparameters}\label{subsec:cv}

	The proposed Bayesian method has hyperparameters that need to be determined, and we provide the following guidelines.
	It is reasonable to use the hyperparameter $\alpha$ close to 1 unless there is a strong evidence of model misspecification.
	Thus, in our numerical studies, $\alpha =0.99$ is used.
	We use the hyperparameter $\nu_0=0$, which leads to the Jeffreys' prior \citep{jeffreys1946invariant} for $d_j$.
	The upper bound for model sizes, $R_j$, is set at $R_j = \lfloor n /2 \rfloor - 2$.
	The rest of hyperparameters $\gamma, c_1$ and $c_2$ control the penalty for large models through $\pi_\alpha(k_j \mid \bfX_n)$: as the values of $\gamma, c_1$ and $c_2$ increase,  $\pi_\alpha(k_j \mid \bfX_n)$ prefers smaller values of $k_j$.
	We suggest to determine $c_2$ using the Bayesian cross-validation method \citep{gelman2014understanding}, with the other hyperparameters $\gamma =0.1$ and $c_1=1$ fixed.
	Compared with the Bayesian cross-validation method for choosing $(\gamma, c_1,c_2)$, this approach significantly reduces computational time, while achieving nice performance in our simulation study.

	To conduct Bayesian cross-validation, we repeatedly split the data $n_{\rm cv}$ times into a training set and a test set with size $n_1 = \lceil n/2 \rceil$ and $n_2 = \lfloor n/2 \rfloor$, respectively.
	Let $I_1(\nu)$ and $I_2(\nu)$ be indices for the $\nu$th training set and test set, respectively, i.e., $|I_1(\nu)| = n_1$, $|I_2(\nu)|=n_2$ and $I_1(\nu) \cup I_2(\nu)= \{1,\ldots,n \}$, for any $\nu = 1,\ldots, n_{\rm cv}$.
	Denote $\bfX_{I_1(\nu)}$ and $\bfX_{I_2(\nu)}$ as $\{X_i\}_{i \in I_1(\nu)}$ and $\{X_i\}_{i \in I_2(\nu)}$, respectively.
	Then for a given hyperparameter $c_2$, the estimated out-of-sample log predictive density, $\text{lpd}_{\rm cv, c_2}$, is 	
	\bea
	\text{lpd}_{\rm cv, c_2} &=& \sum_{\nu =1}^{n_{\rm cv}} \log f_{c_2} ( \bfX_{I_2(\nu)} \mid  \bfX_{I_1(\nu)}  ) \\
	&=&  \sum_{\nu =1}^{n_{\rm cv}} \log   \Big\{ \sum_{k}  f (\bfX_{I_2(\nu)} \mid k) \pi_{\alpha,c_2} (k \mid \bfX_{I_1(\nu)}  )  \Big\}   ,
	\eea
	where $k = (k_2,\ldots, k_p)$, $ f (\bfX_{I_2(\nu)} \mid k) $ is the marginal likelihood for $k$ given $\bfX_{I_2(\nu)}$, and $\pi_{\alpha,c_2} (k \mid \bfX_{I_1(\nu)}  ) $ is the fractional posterior based on $\bfX_{I_1(\nu)}$ and the hyperparameter $c_2$.
	The aim of the Bayesian cross-validation is to find the optimal $c_2$ maximizing $\text{lpd}_{\rm cv, c_2}$.

	The marginal likelihood $ f (\bfX_{I_2(\nu)} \mid k) $ is available in a closed form:
	\bea
	f (\bfX_{I_2(\nu)} \mid k) \,=\, \prod_{j=2}^p  \Bigg[ (2\pi )^{-n_2/2} \Gamma \Big( \frac{n_2 + \nu_0}{2} \Big) \Big( 1+ \frac{1}{\gamma} \Big)^{-k_j /2} \Big\{  \what{d}_j^{(k_j)}(I_2(\nu)) /2 \Big\}^{ - (n_2 + \nu_0)/2 }  \Bigg]  ,
	\eea
	where $\what{d}_j^{(k_j)}(I_2(\nu))$ is the estimated variance $\what{d}_j^{(k_j)}$ using $\bfX_{I_2(\nu)}$.
	The closed form of the marginal posterior $\pi_{\alpha,c_2} (k_j \mid \bfX_{I_1(\nu)}  ) $ is also available in \eqref{posteriors}, which requires the calculation of $\what{d}_j^{(k_j)}(I_1(\nu))$. 
	When calculating $\text{lpd}_{\rm cv, c_2}$, the main computational burden comes from calculating $\what{d}_j^{(k_j)}(I_1(\nu))$ and $\what{d}_j^{(k_j)}(I_2(\nu))$ for each $k_j = 0,1,\ldots, R_j \wedge (j-1)$ and $j=2,\ldots, p$.
	Note that these quantities do not vary from  different choices of $c_2$. 
	For a given randomly split data, these quantities only need to be calculated once regardless of the value of $c_2$.
	Therefore, LANCE prior enables scalable cross-validation-based inference even in high-dimensions.
	Throughout the numerical study, we split the data $n_{\rm cv} = 5$ times.

	\subsection{Simulated data}\label{subsec:simulated}
	
	\bch Throughout the numerical studies in this section, we focus on the bandwidth selection performance, while additional simulation studies focusing on the estimation performance are given in the Appendix.\ech
	We generated the true precision matrix $\Omega_{0n} = (I_p - A_{0n})^T D_{0n}^{-1} (I_p-A_{0n})$ for simulation studies.
	The diagonal entries of $D_{0n} = diag( d_{0j} )$ were drawn independently from $Unif(2,5)$.
	Next, the lower triangular entries of $A_{0n} = ( a_{0, jl} )$ were generated as follows:
	\begin{itemize}
		\item Model 1: 
		\bch For $2\le v \le p$, the bandwidth of the $v$th row of $A_{0n}$ is sampled from $Unif\{ 1, \ldots, \min( v-1, 5) \}$.
		This produces a sparse Cholesky factor with the maximum bandwidth size 5.\ech
		Each nonzero element in $A_{0n}$ is drawn independently from $Unif(A_{0,\min},  A_{0, \max} )$, where the positive of negative sign is assigned with probability 0.5.
		
		\item Model 2: $A_{0n}$ is a block diagonal matrix consisting of 5 blocks with size $p/5$, while the maximum size of bandwidths is 40.
		This setting produces a moderately sparse Cholesky factor.
		The length of the bandwidth of the $v$th row in each block follows a mixture distribution, $ 0.5 \times Unif\{1, \ldots, \min (v-1 , 40 )\} + 0.5 \times \delta_0 $, where $\delta_0$ is a point mass at zero.
		Each nonzero element in $A_{0n}$ is drawn independently from $Unif(A_{0,\min},  A_{0, \max} )$, where the positive or negative sign is assigned with probability 0.5.
		This setting corresponds to Model 2 in \cite{yu2017learning}.
		
		\item Model 3: $A_{0n}$ is a block diagonal matrix consisting of 2 blocks with size $p/2$, and the rest  of the generation process is similar to  that of Model 2.
		This setting produces a denser Cholesky factor compared with Model 2.
		This setting corresponds to Model 3 in \cite{yu2017learning}.
	\end{itemize}		
	Figure \ref{fig:A0s} shows a simulated true Cholesky factors for Settings 1 and 2 with $(A_{0,\min}, A_{0,\max}) =(0.1,  0.4)$ and $p=100$.
	\begin{figure*}[!tb]
		\centering
		\includegraphics[width=13.6cm,height=4.5cm]{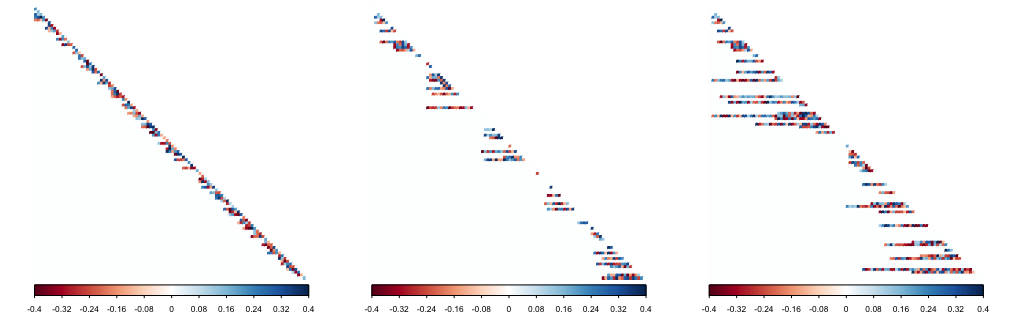}
		\caption{
			The true Cholesky factors for Model 1 (left), Model 2 (middle) and Model 3 (right) with $(A_{0,\min}, A_{0,\max}) =(0.1,  0.4)$ and $p=100$.
		}
		\label{fig:A0s}
	\end{figure*}

	We compare the performance of the proposed LANCE prior with the penalized likelihood approach in \cite{yu2017learning}, which we call YB method.
	Since the unweighted version outperformed the weighted version in the simulation studies in \cite{yu2017learning}, we used the unweighted version of the penalized likelihood approach.
	Furthermore, the convex sparse Cholesky selection (CSCS) \citep{khare2019scalable}, which is designed for sparse Cholesky factors, is also considered as a contender.
	Let $\what{A}_n^{YB}$ and $\what{A}_n^{CSCS}$ be the estimated Cholesky factor based on the penalized likelihood approach proposed by \cite{yu2017learning} and \cite{khare2019scalable}, respectively.
	To estimate the local dependence structure, we set all of the estimated entries in $\what{A}_n^{YB}$ whose absolute values are below $0.1^{10}$ to zero, as suggested by \cite{yu2017learning}.
	The receiver operating characteristic (ROC) curves for LANCE prior and the penalized likelihood approaches were drawn based on 100 hyperparameters $c_2$ selected from $[-1.5, 5]$ and 100 tuning parameters $\lambda$ selected from $[0.01, 4]$, respectively.
	Using these hyperparameter values, we also compare the performance of cross-validation for LANCE prior and YB method.
	The Bayesian cross-validation described in Section \ref{subsec:cv} was used for LANCE prior.
	For YB method, we used \verb|varband_cv| function in \verb|R| package \verb|varband|. 
	The $5$-fold cross-validation was used based on the unweighted version of the penalty.
	We did not conduct a cross-validation for  CSCS method due to heavy computation.
	\bch The three methods, the LANCE, YB and CSCS methods, can be run in parallel, although we did not use parallel computing in the numerical study.\ech

	\begin{figure*}[!tb]
		\centering
		\includegraphics[width=5.5cm,height=4.cm]{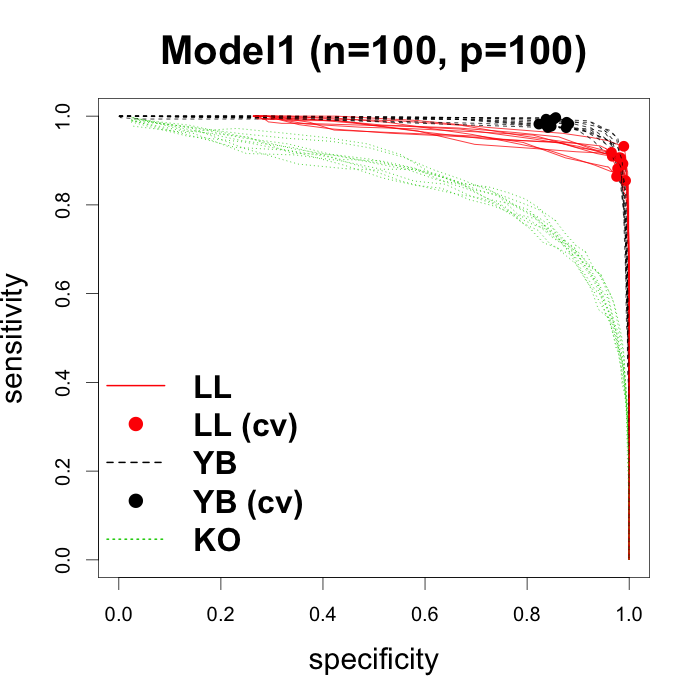}
		\includegraphics[width=5.5cm,height=4.cm]{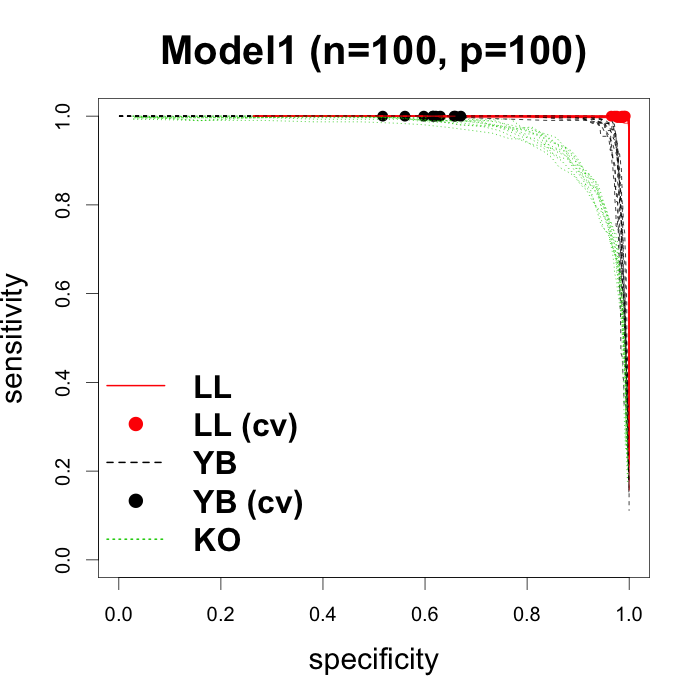}
		\includegraphics[width=5.5cm,height=4.cm]{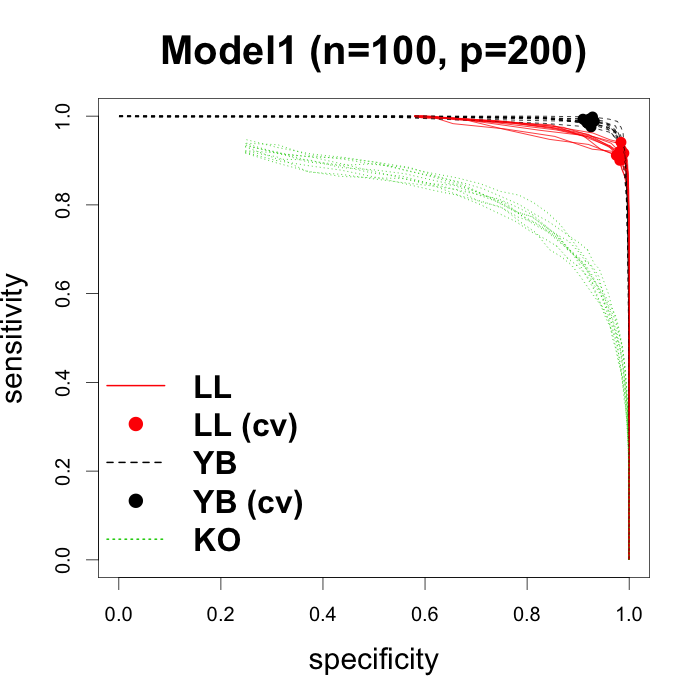}
		\includegraphics[width=5.5cm,height=4.cm]{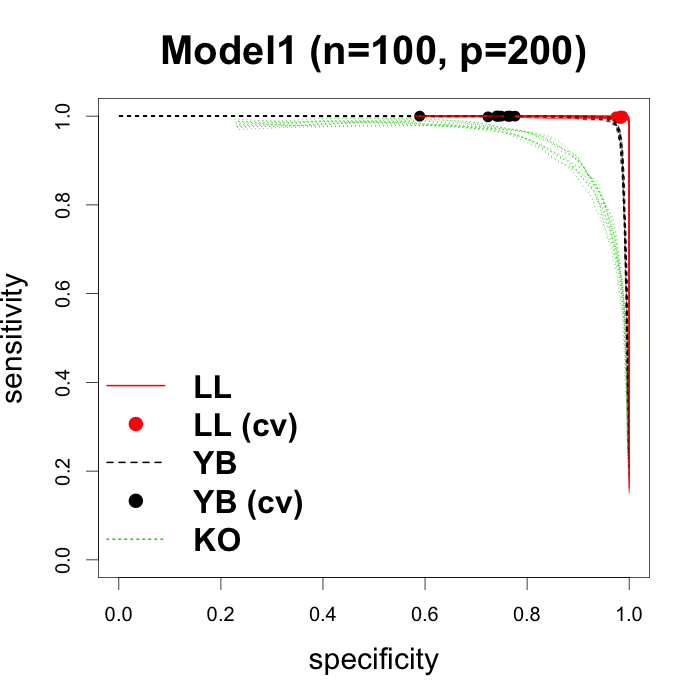}
		\includegraphics[width=5.5cm,height=4.cm]{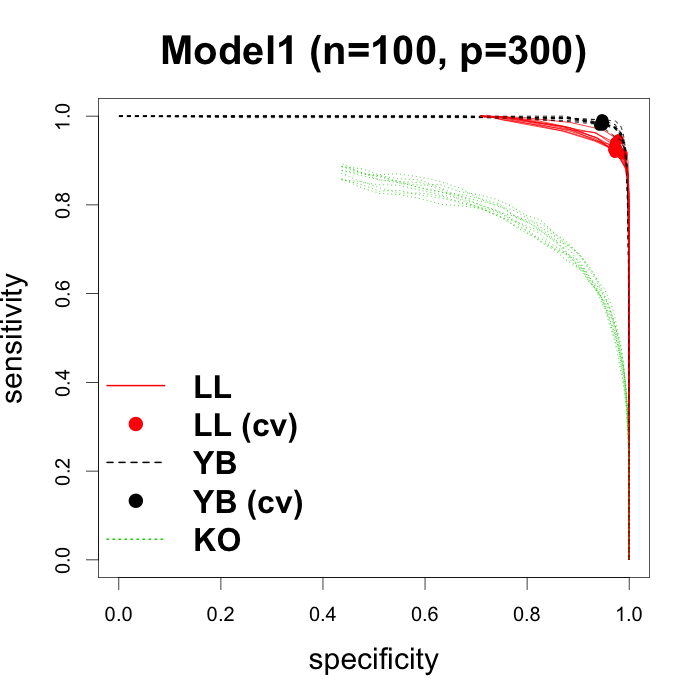}
		\includegraphics[width=5.5cm,height=4.cm]{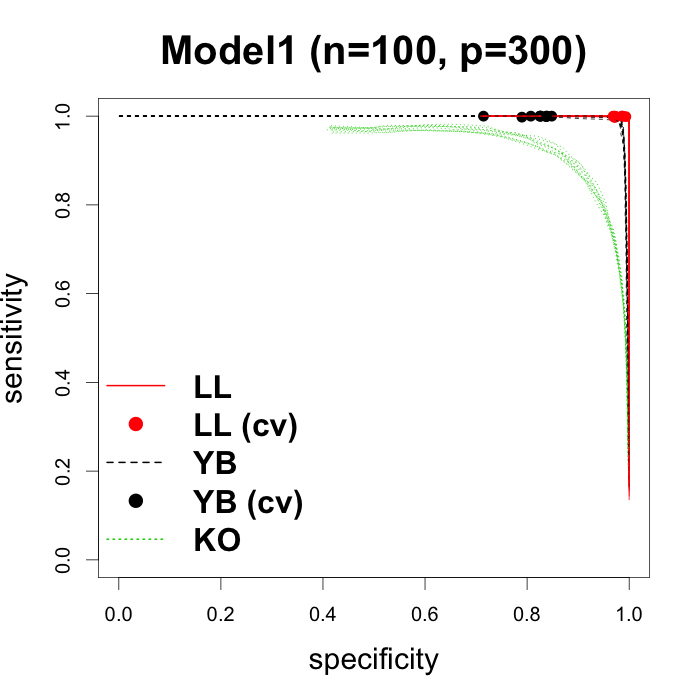}
		\vspace{-.2cm}
		\caption{
			ROC curves are represented based on 10 simulated data sets from Model 1 with $n=100$ and $p \in \{100, 200, 300\}$.  
			Left column and right column show the results for $(A_{0,\min}, A_{0,\max}) =(0.1,  0.4)$ and $(A_{0,\min}, A_{0,\max}) =(0.4,  0.6)$, respectively.
			LL and YB represent the methods proposed in this paper and \cite{yu2017learning}, respectively.
		}
		\label{fig:roc1}
	\end{figure*}	
	\begin{figure*}[!tb]
		\centering
		\includegraphics[width=5.5cm,height=4.cm]{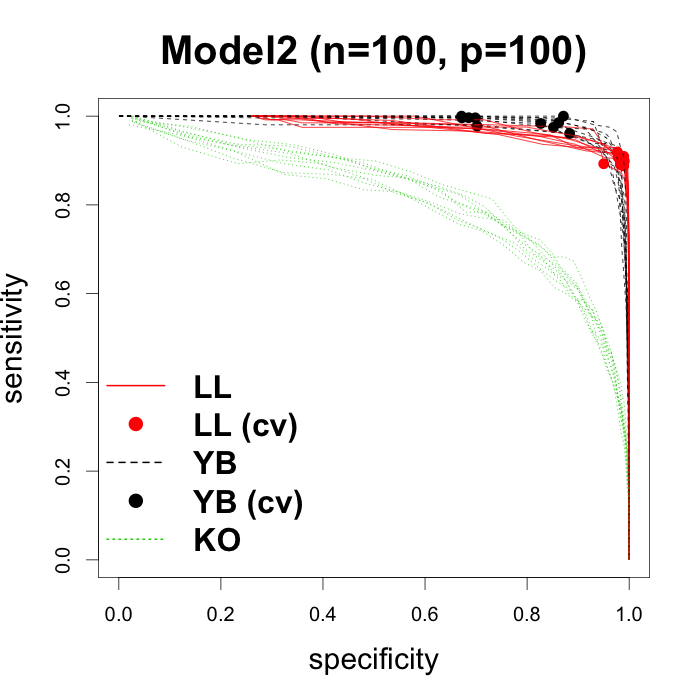}
		\includegraphics[width=5.5cm,height=4.cm]{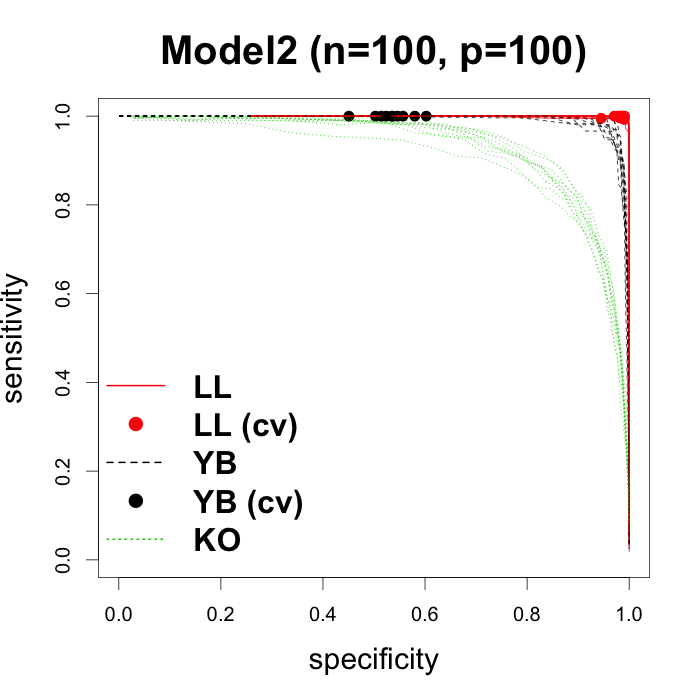}
		\includegraphics[width=5.5cm,height=4.cm]{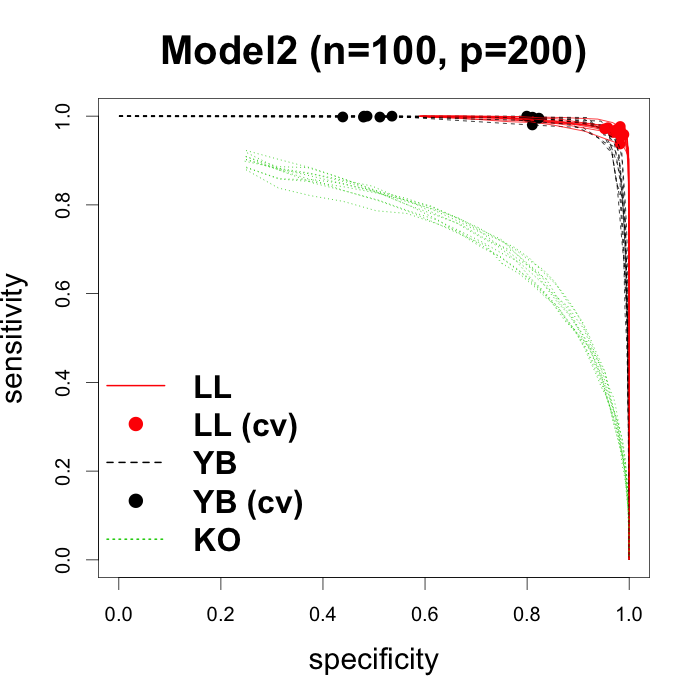}
		\includegraphics[width=5.5cm,height=4.cm]{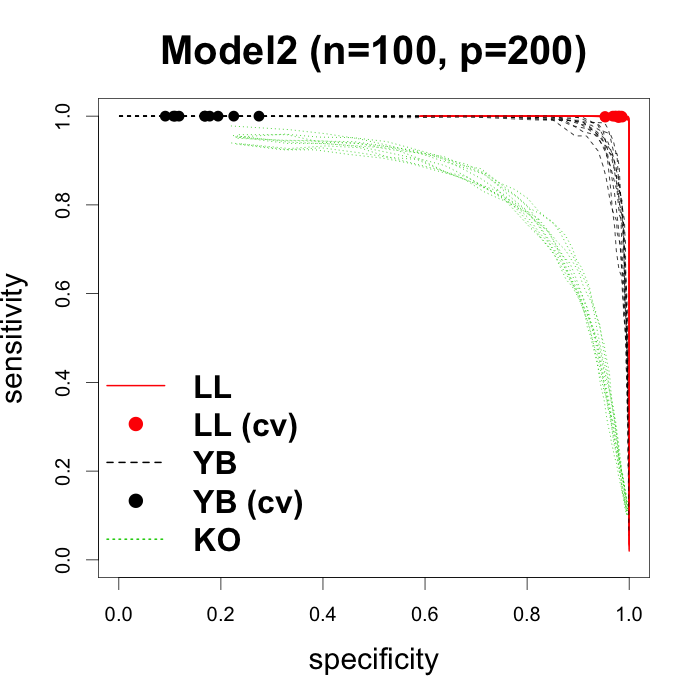}
		\includegraphics[width=5.5cm,height=4.cm]{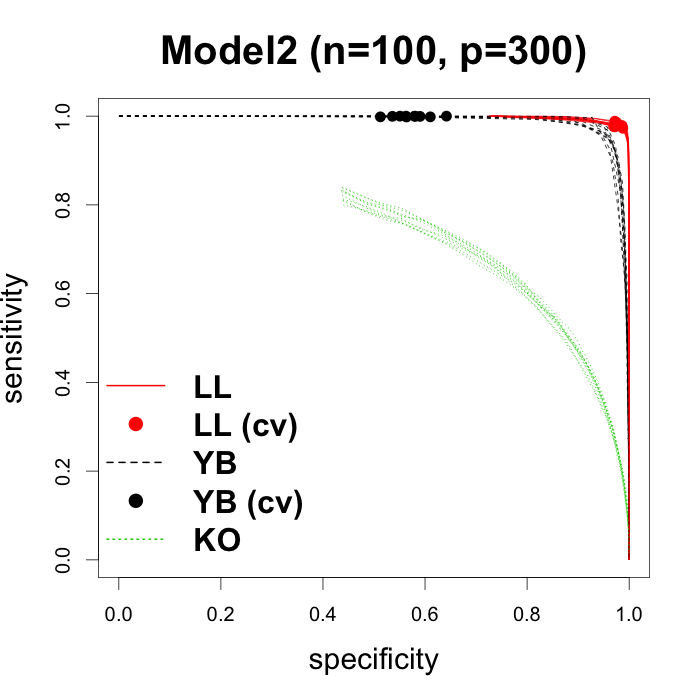}		
		\includegraphics[width=5.5cm,height=4.cm]{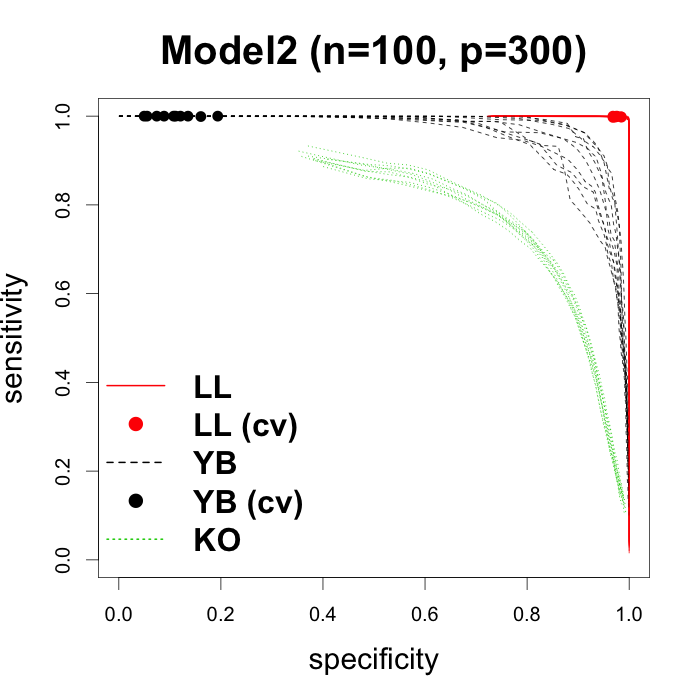}
		\vspace{-.2cm}
		\caption{
			ROC curves are represented based on 10 simulated data sets from Model 2 with $n=100$ and $p \in \{100, 200, 300\}$.  
			Left column and right column show the results for $(A_{0,\min}, A_{0,\max}) =(0.1,  0.4)$ and $(A_{0,\min}, A_{0,\max}) =(0.4,  0.6)$, respectively.
		}
		\label{fig:roc2}
	\end{figure*}
	\begin{figure*}[!tb]
		\centering
		\includegraphics[width=5.5cm,height=4.cm]{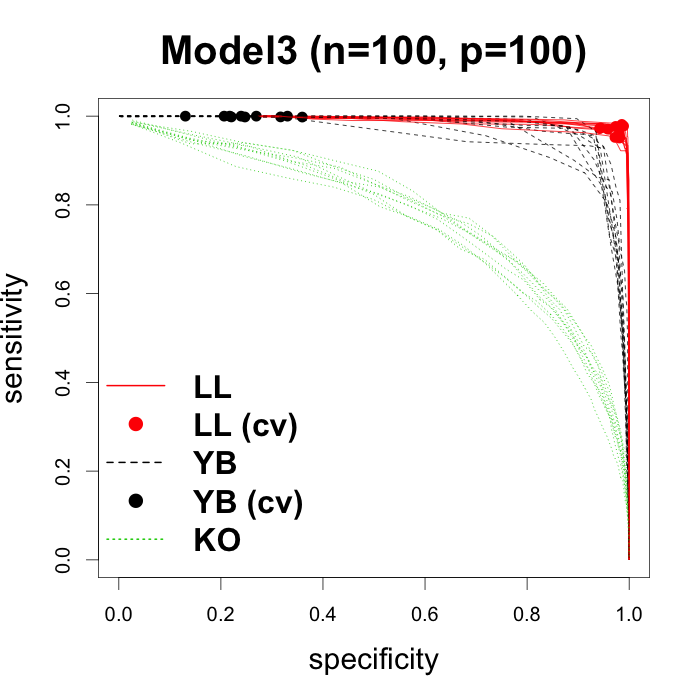}
		\includegraphics[width=5.5cm,height=4.cm]{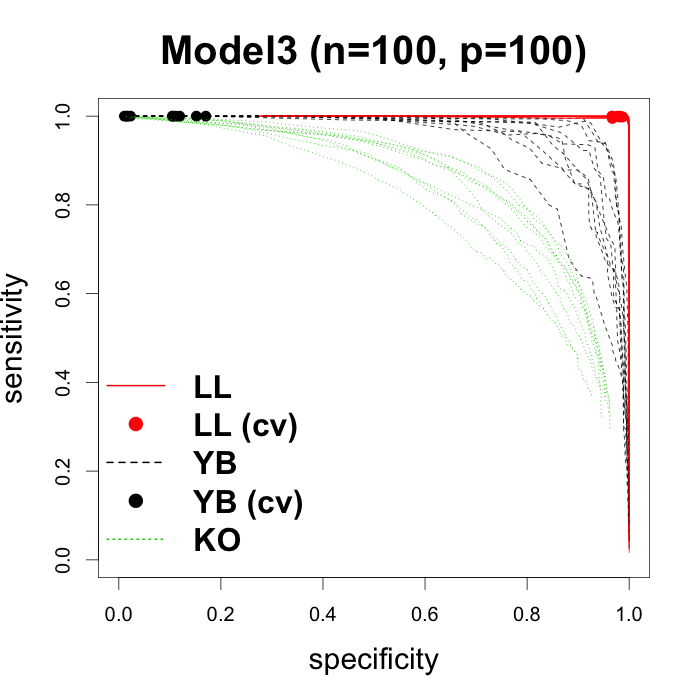}
		\includegraphics[width=5.5cm,height=4.cm]{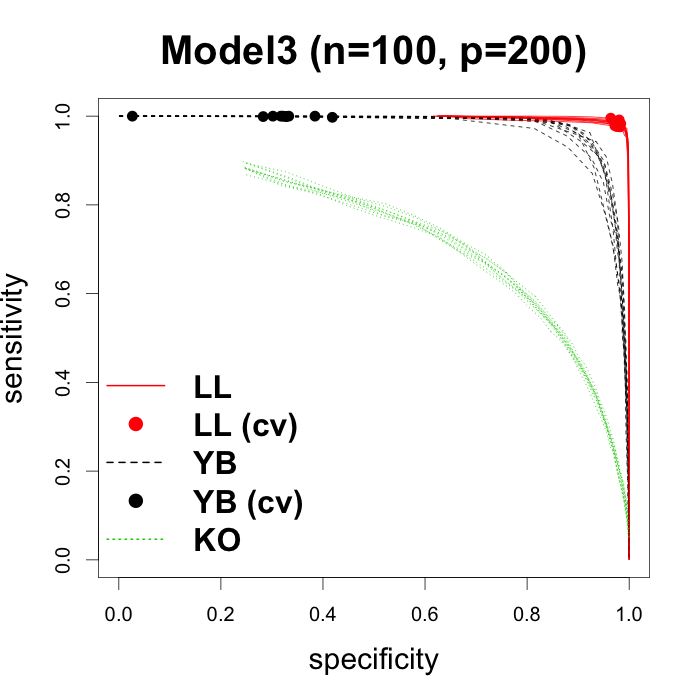}
		\includegraphics[width=5.5cm,height=4.cm]{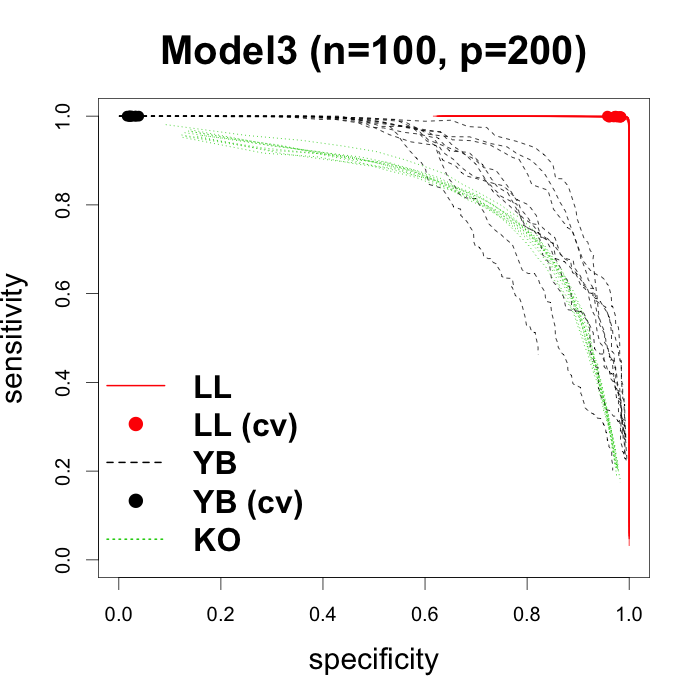}
		\includegraphics[width=5.5cm,height=4.cm]{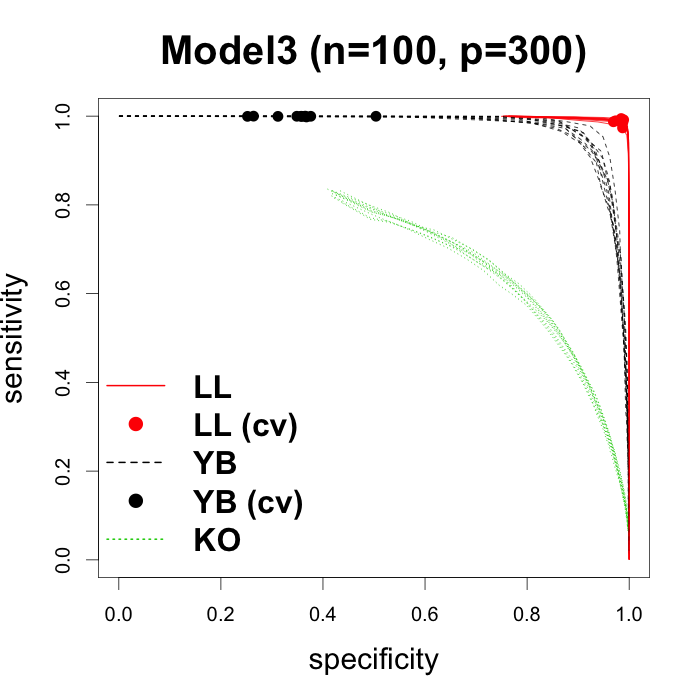}
		\includegraphics[width=5.5cm,height=4.cm]{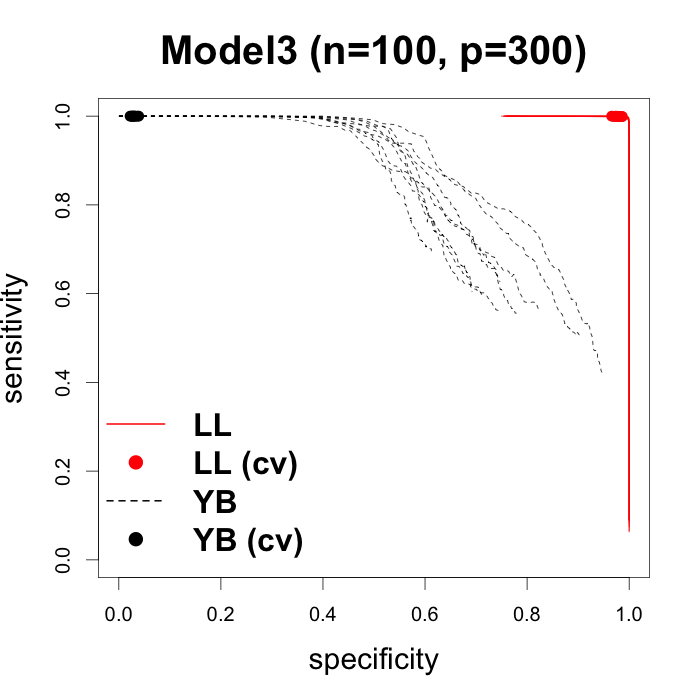}
		\vspace{-.2cm}
		\caption{
			ROC curves are represented based on 10 simulated data sets from Model 3 with $n=100$ and $p \in \{100, 200, 300\}$.  
			Left column and right column show the results for $(A_{0,\min}, A_{0,\max}) =(0.1,  0.4)$ and $(A_{0,\min}, A_{0,\max}) =(0.4,  0.6)$, respectively.
		}
		\label{fig:roc3}
	\end{figure*}
	
	
	Figures \ref{fig:roc1}, \ref{fig:roc2} and \ref{fig:roc3} represent ROC curves based on 10 simulated data sets for Model 1, Model 2 and Model 3, respectively.
	For Model 3 with $p=300$, we omit the results for CSCS method, because it did not converge for several days and caused a convergence problem.
	\bch As expected, CSCS method does not work well compared with other two methods tailored to local dependence structure.
	The main reason for this phenomenon is that CSCS method does not guarantee local dependence structure.
	Based on the simulation results, the performances of LANCE prior and YB method are comparable in Model 1 (i.e., when bandwidths are smaller than $5$), but LANCE prior tends to give larger area under the curves than those of YB method in Model 2 (i.e., when bandwidths are moderately large).
	Especially in Model 3 (i.e., when bandwidths are large), LANCE prior significantly outperforms YB method especially for large $p$.
	Thus, it seems that YB method tends to work better with smaller bandwidths, which is consistent with the observations in \cite{yu2017learning}.\ech
	Furthermore, we found that LANCE prior works better under large signals, $(A_{0,\min}, A_{0,\max}) =(0.4,  0.6)$, compared with small signals, $(A_{0,\min}, A_{0,\max}) =(0.1,  0.4)$.
	This makes sense because it is expected that large signals will easily satisfy the beta-min condition (A2).

	The dots in Figures \ref{fig:roc1}, \ref{fig:roc2} and \ref{fig:roc3} show the results based on cross-validation, where red dots and black dots represent those of LANCE prior and YB method, respectively.
	We found that cross-validation based on LANCE prior gives nearly optimal result in the sense that the result of the cross-validation method is located close to $(1,1)$ on the ROC curve.
	On the other hand, cross-validation based on YB method tends to produce high false positive, which results in low specificity.
	In many cases, even when ROC curves of the two methods are similar, the performances of our cross-validation-based inference are much better than those of YB method.

	\begin{figure*}[!tb]
		\centering
		\includegraphics[width=5.5cm,height=4.cm]{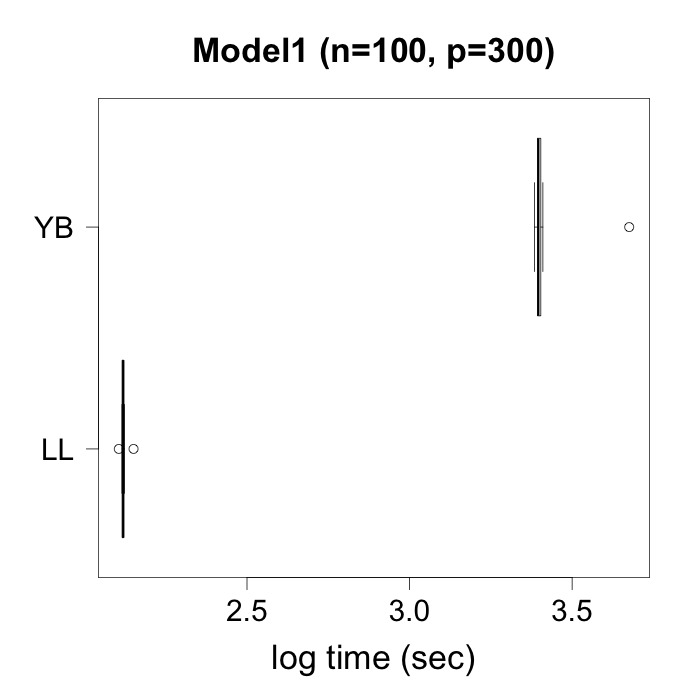}
		\includegraphics[width=5.5cm,height=4.cm]{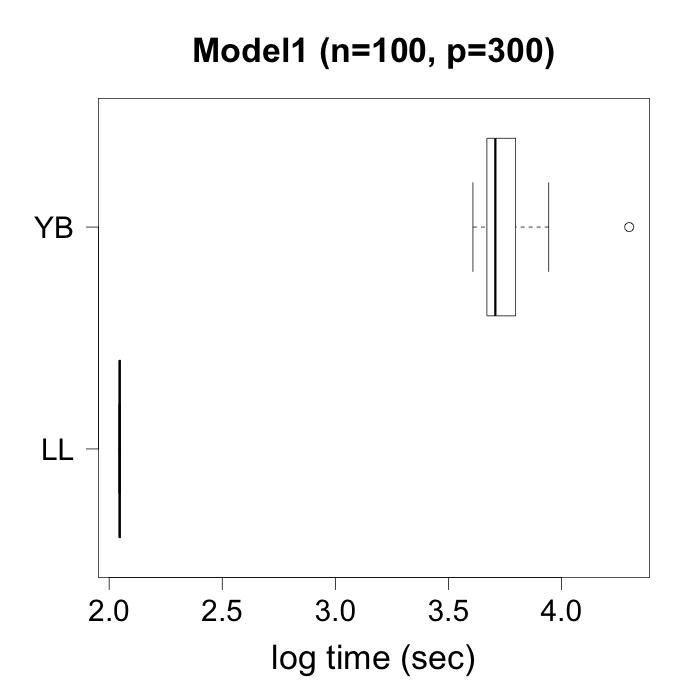}
		\includegraphics[width=5.5cm,height=4.cm]{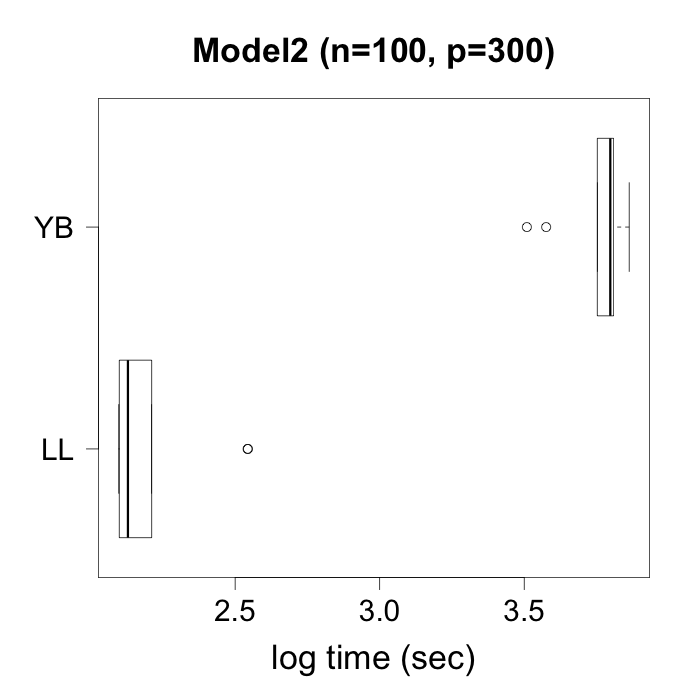}
		\includegraphics[width=5.5cm,height=4.cm]{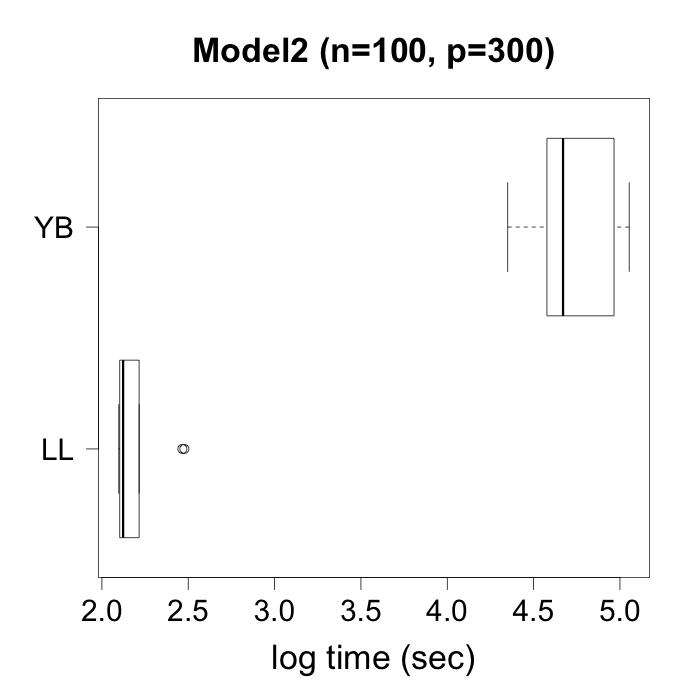}
		\includegraphics[width=5.5cm,height=4.cm]{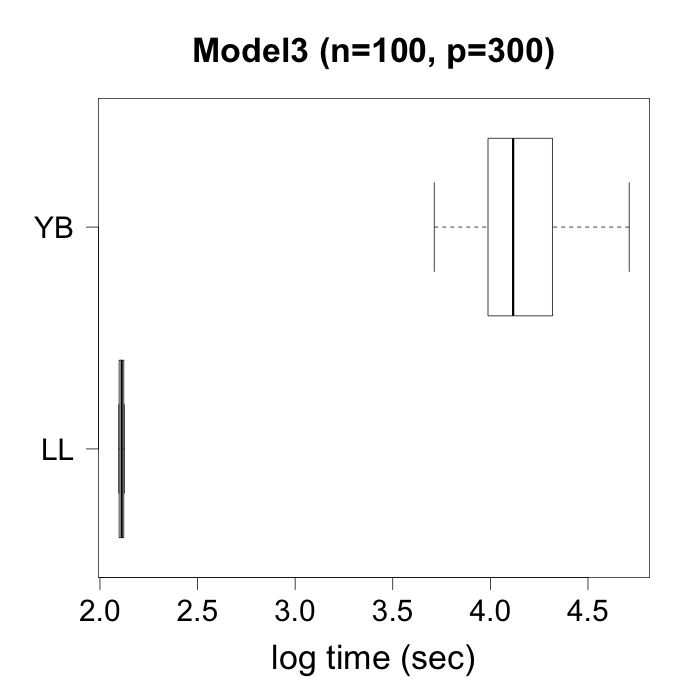}
		\includegraphics[width=5.5cm,height=4.cm]{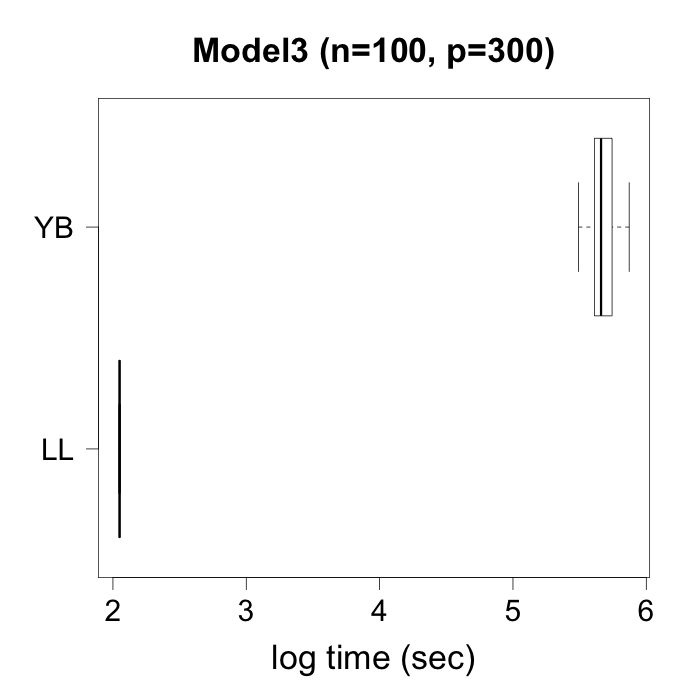}
		\vspace{-.2cm}
		\caption{
			Logarithm of computation times for each method based on 10 simulated data sets.
			Left column and right column show the results for $(A_{0,\min}, A_{0,\max}) =(0.1,  0.4)$ and $(A_{0,\min}, A_{0,\max}) =(0.4,  0.6)$, respectively.
		}
		\label{fig:time}
	\end{figure*}
	We also found that the proposed method is much faster than YB method in most settings.
	Figure \ref{fig:time} shows box plots of the computation times for cross-validation using 100 hyperparameters.
	For each method, box plots were drawn based on 10 simulated data sets with $n=100$ and $p=300$.
	The relative computational gain of  LANCE prior can be summarized by dividing the computation time for the LANCE prior by that for  YB method.
	In our simulation settings, the mean and median of the relative computational gain of  LANCE prior were 1539 and 208, respectively, which clearly show a computational advantage of the proposed method.
	Also note that YB method was conducted using the \verb|R| package \verb|varband| providing \verb|C++| implementations, while  LANCE prior was implemented using only \verb|R|.
	The main reason for this observation is that YB method requires solving a penalized likelihood problem for each value of the tuning parameter $\lambda$.
	On the other hand, for  LANCE prior, once the estimated error variance, $\what{d}_j^{(k_j)}$, is calculated, there is no need to recalculate it for various values of $c_2$.
	As a result, the cross-validation for LANCE prior is much faster than the state-of-the-art contender, thus enables scalable inference even in high-dimensions.

	\subsection{Real data analysis: phone call center and gun point data}\label{subsec:real}
	
	We  demonstrate the practical performance of  LANCE prior by applying our model to two real data examples.  We first consider the  telephone call center data set, which was analyzed by \cite{huang2006covariance} and \cite{bickel2008regularized}.
	This data set consists of phone calls for 239 days in 2002 from a call center of a major U.S. financial organization. 
	The phone calls were recorded from 7:00 am until midnight for every 10 minutes, resulting in 102 intervals for each day, except holidays, weekends and days when the recording system did not work properly.
	The number of calls on the $j$th time interval of the $i$th day is denoted by $N_{ij}$ for $i=1,\ldots, 239$ and $j=1,\ldots, 102$ $(p=102)$.
	The first $205$ days are used as a training set $(n=205)$, and the last $34$ days are used a test set.
	The data were transformed to $X_{ij} = \sqrt{N_{ij} + 1/4 }$ as in \cite{huang2006covariance} and \cite{bickel2008regularized}.
	The data were centered after the transformation.
	Note that the data has a natural time ordering between the variables making it appropriate to apply  LANCE prior.

	The primary purpose is to predict the number of phone calls during a certain time period.
	We divide 102 time intervals for each day into two groups: those before the 51st interval and those after 52nd interval.
	For each $j=52,\ldots, 102$, we predict $X_{ij}$ using the best linear predictor based on $X^j := (X_{i1},\ldots, X_{i, j-1})^T$, 
	\bea
	\hat{X}_{ij}  &=&  \mu_j  + \Sigma_{(j, 1:(j-1) )} \{ \Sigma_{(1:(j-1), 1:(j-1) )}  \}^{-1}  ( X^j - \mu^j ) ,
	\eea
	where $\mu_j = \bbE(X_{1j})$, $\mu^j = (\mu_1,\ldots, \mu_{j-1})^T$ and $\Sigma_{S_1, S_2} = ( \sigma_{ij}  )_{i \in S_1, j \in S_2 }$.
	The unknown parameters are estimated by $\hat{\mu}_j = \sum_{i=1}^{205} X_{ij}/205$ and $\hat{\Sigma} = \hat{\Omega}^{-1}$, where $\hat{\Omega}$ are estimated using various methods including  LANCE prior: LANCE prior, YB method \citep{yu2017learning}, ESC prior \citep{lee2019minimax} and CSCS method \citep{khare2019scalable}.
	\bch For LANCE prior, $\what{a}_j^{(\hat{k}_j) }$ and $\what{d}_j^{(\hat{k}_j) }$ are used to construct $\hat{\Omega}$, where $\hat{k}_j$ is the posterior mode.
	For ESC prior, instead of varying bandwidths, the posterior sample-based mode of the support of the Cholesky factor is used.\ech

	The absolute prediction error is calculated by $PE_j = \sum_{i=206}^{239} |X_{ij} - \hat{X}_{ij} |/34$ for each $j=52,\ldots, 102$, and the average of prediction errors, $\sum_{j=52}^{102} PE_j$, is used to evaluate the performance of each method.
	The hyperparameter in each method is chosen based on cross-validation, except ESC prior.
	Because applying cross-validation to  ESC prior is prohibitive due to heavy computation, we set the hyperparameters in  ESC prior at $\gamma=0.1$, $\nu_0=0$, $c_1=0.0005$ and $c_2=1$ as suggested by \cite{lee2019minimax}.
	Figure \ref{fig:call_cencer} represents prediction errors at each time point.
	Averages of prediction errors for the methods proposed in this paper, \cite{yu2017learning}, \cite{lee2019minimax} and \cite{khare2019scalable} are $0.5502$, $0.5531$, $0.5576$ and $0.5708$, respectively.
	It suggests that a local dependence structure is more suitable for the call center data than an arbitrary dependence structure, which makes sense due to the nature of the data.
	
	\begin{figure*}[!tb]
		\centering
		\includegraphics[width=13.cm,height=10.cm]{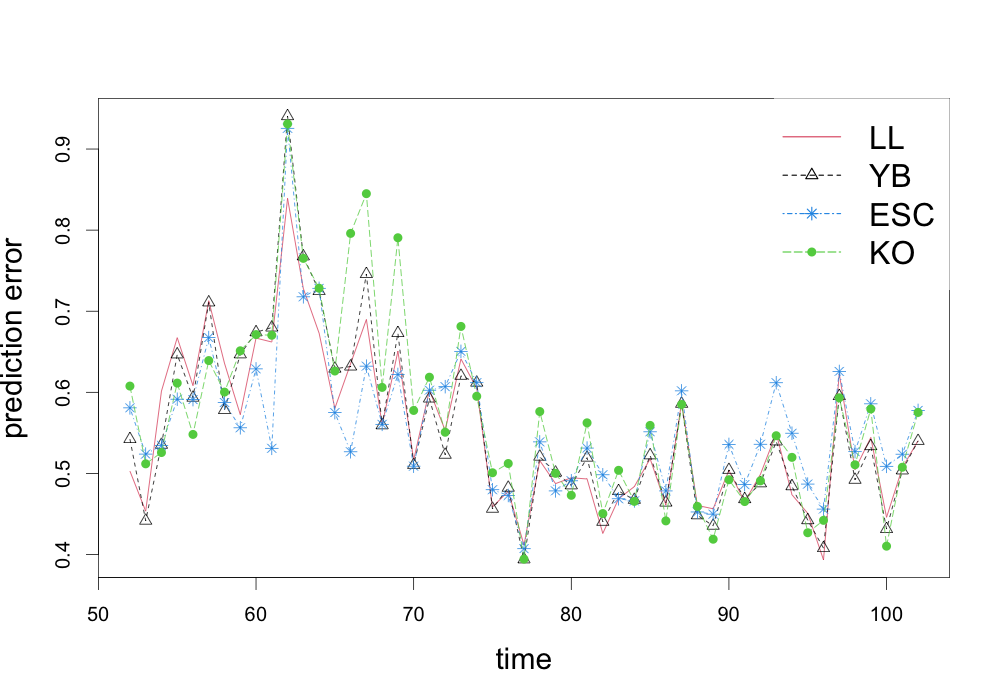}
		\vspace{-.2cm}
		\caption{
			Prediction errors at each time point for each method.
			LL, YB, ESC and KO represent the methods proposed in this paper, \cite{yu2017learning}, \cite{lee2019minimax} and \cite{khare2019scalable}, respectively.
		}
		\label{fig:call_cencer}
	\end{figure*}

	%
	%

	We further illustrate the performance of  LANCE prior in a classification problem.
	The GunPointAgeSpan data set, which is publicly available at \url{http://timeseriesclassification.com}, is used to conduct the quadratic discriminant analysis (QDA). 
	This consists of the two GunPoint data sets released in 2003 and 2018, respectively, and each year, the same two actors (one male and one female) participated in the experiment.
	There are two classes in a data set: Gun and Point.
	For the Gun class, the actors hold a gun and point the gun toward a target, while they point with just their fingers (without a gun) in the Point class.
	\bch The x-axis coordinates of centroid of the hand are recorded over five seconds of the movement based on 30 frames per second, which results in 150 frames $(p=150)$ per action.\ech
	Thus, this data set also has a time ordering between the variables.
	The GunPointAgeSpan has 135 training set $(n=135)$ and 316 test set, and the purpose of the analysis is to classify test observations into the two classes (Gun and Point).
	The numbers of observations corresponding to the Gun class are 68 and 160 in the training and test data, respectively.
	The training and test data sets are centered.

	For each data $x$ in the test set, the quadratic discriminant score $\delta_k(x)$ is calculated as follows:
	\bea
	\delta_k(x)  &=& \frac{1}{2} \log \det ( \Omega_k )  - \frac{1}{2} (x -  {\mu}_k)^T  \Omega_k (x-  {\mu}_k) + \log \Big( \frac{n_k}{n} \Big), \quad k=1,2, 
	\eea
	where $\mu_k$ and $\Omega_k$ are the mean vector and precision matrix for the class $k=1,2$, respectively.
	Here, $n_k$ is the number of observations for the class $k$, thus we have $n_1=68$ and $n_2=67$.
	To conduct the QDA, we estimate the unknown parameters $\mu_k$ and $\Omega_k$ using the training data set. 
	They are plugged into $\delta_k(x)$ similarly to the phone call center data example, and $x$ is then classified as the class $\hat{k} = \argmax_k \delta_k(x)$.
	The performances of  LANCE prior, YB method, ESC prior and CSCS method are compared, where cross-validation is used for each method except  ESC prior.

	\begin{table}[!tb]
		\centering
		\caption{Classification errors for the test set in the GunPointAgeSpan data.}
		\begin{tabular}{l | rrrr}
			\hline
			& LL & YB &  ESC & KO \\ \hline
			Error & 0.2310  & 0.3956 & 0.2437 & 0.3704 \\ 
			\hline
		\end{tabular}
		\label{table:class_error}
	\end{table}
	Table \ref{table:class_error} represents classification errors for the test set based the QDA using $\hat{\Omega}_k$ estimated by each method.
	For this data set, Bayesian methods seem to achieve lower classification errors than the penalized likelihood approaches, while  LANCE prior achieves the lowest classification error.
	Despite similar performance,  LANCE prior has a clear advantage over  ESC prior by enabling a scalable cross-validation to select the hyperparameter.
	In practice, there is no guideline for choosing the hyperparameters in  ESC prior, which can dramatically affect the performance.

	\section{Discussion}\label{sec_disc}
	
	In this paper, we propose a Bayesian procedure for high-dimensional local dependence learning, where variables close to each other are more likely to be correlated.
	The proposed prior, LANCE prior, allows an exact computation of posteriors, which enables scalable inference even in high-dimensional settings.
	Furthermore, it provides a scalable Bayesian cross-validation to choose the hyperparameters.
	We establish selection consistency for the local dependence structure and posterior convergence rates for the Cholesky factor.
	The required conditions for these theoretical results are significantly weakened compared with the existing literature.
	Simulation studies in various settings show that  LANCE prior outperforms other contenders in terms of the ROC curve, cross-validation-based analysis and computation time. 
	Two real data analyses based on the phone call center and gun point data illustrate the satisfactory performance of the proposed method in linear prediction and classification problems, respectively. 	
	
	\bch 
	It is worth mentioning that LANCE prior is only applicable when $k_j$ is smaller than $n$ due to $(\bfX_{j (k_j)}^T \bfX_{j (k_j)})^{-1}$ term in the conditional prior for $a_j^{(k_j)}$.
	Although we rule out this situation by introducing condition (A3), in practice, we can modify LANCE prior to allow $k_j >n$  when needed.
	Specifically, we can modify the conditional prior for $a_j^{(k_j)}$ to 
	\bean\label{mod_a_prior}
	a_j^{(k_j)} \mid d_j, k_j &\overset{ind.}{\sim}& N_{k_j} \Big(  \tilde{a}_j^{(k_j)} ,    \frac{d_j}{\gamma}\big( \bfX_{j (k_j)}^T \bfX_{j (k_j)} + c I_{k_j} \big)^{-1}     \Big) , \,\, j=2,\ldots, p,
	\eean
	for some constant $c>0$, where $\tilde{a}_j^{(k_j)}  = \big( \bfX_{j (k_j)}^T \bfX_{j (k_j)} + c I_{k_j}  \big)^{-1} \bfX_{j (k_j)}^T \tilde{X}_j$.
	Note that the above prior \eqref{mod_a_prior} has a variance stabilizing term $c I_{k_j}$.
	Then, the resulting $\alpha$-fractional posterior has the following closed form:
	\bea
	\begin{split}
		a_{j}^{(k_j)} \mid d_j, k_j, \bfX_n \,\,&\overset{ind.}{\sim}\,\,  N_{k_j} \Big(  \big\{  \bfX_{j (k_j)}^T \bfX_{j (k_j)} +  \frac{c \gamma}{\alpha+\gamma}  I_{k_j}  \big\}^{-1} \bfX_{j (k_j)}^T \tilde{X}_j  ,  \\
		& \hspace{2.7cm}  d_j \big\{ (\alpha+\gamma) \bfX_{j (k_j)}^T \bfX_{j (k_j)} +  c \gamma I_{k_j}  \big\}^{-1}    \Big) , \,\, j=2,\ldots, p,   \\
		d_j \mid k_j , \bfX_n \,\,&\overset{ind.}{\sim}\,\, IG \Big(  \frac{\alpha n + \nu_0}{2} , \, \frac{\alpha n}{2}\tilde{d}_{j}^{(k_j)}  \Big) , \,\, j=1,\ldots, p , \\
		\pi_\alpha(k_j \mid \bfX_n)  \,\,&\propto\,\,  \pi(k_j) \Big\{ \big( \bfX_{j (k_j)}^T \bfX_{j (k_j)} + c I_{k_j}\big)^{-1} \big( \bfX_{j (k_j)}^T \bfX_{j (k_j)}  + \frac{ c\gamma}{\alpha+\gamma} I_{k_j} \big)   \Big\}^{-\frac{k_j}{2} }   \\
		& \hspace{1.cm} \times  \Big( 1+ \frac{\alpha}{\gamma}  \Big)^{- \frac{k_j}{2}}  \big(  \tilde{d}_{j}^{(k_j)} \big)^{- \frac{\alpha n + \nu_0}{2}} 
		, \,\, j=2,\ldots,p  ,
	\end{split}
	\eea
	where 
	\bea
	\tilde{d}_{j}^{(k_j)} = n^{-1}\tilde{X}_j^T \big\{ I_n  -  \bfX_{j (k_j)} \big( \bfX_{j (k_j)}^T \bfX_{j (k_j)} + \frac{c\gamma}{\alpha+\gamma} I_{k_j} \big)^{-1} \bfX_{j (k_j)}^T \big\} \tilde{X}_j  + \hspace{2cm}  \\
	\tilde{X}_j^T \bfX_{j (k_j)} \big\{ \big( \bfX_{j (k_j)}^T \bfX_{j (k_j)} + \frac{c \gamma}{\alpha+\gamma} I_{k_j}\big)^{-1} - (\bfX_{j (k_j)}^T \bfX_{j (k_j)} + c I_{k_j} )^{-1}    \big\} \bfX_{j (k_j)}^T \tilde{X}_j \, \gamma/ (\alpha n) .
	\eea
	Thus, the proposed LANCE prior can be modified to allow large bandwidths such that $k_j > n$.
	A sufficiently small $c>0$ would give similar results with \eqref{posteriors} in practice, although theoretical properties of the resulting $\alpha$-fractional posterior should be further investigated.
	\ech

	\bch 
	An important and interesting future direction  is to consider the estimation of covariance  or precision structures with an unknown topological ordering on the variables of a DAG.   
	Assume one specifies an arbitrary label of the variables as $x_1,\ldots, x_p$ which does not necessarily correspond to the ordered variables.  
	Let $\sigma \in \mathcal S_p$ be an element in the permutation group of $p$ variables and $G_{\sigma}$ is the corresponding true DAG structure.  
	If one is interested in learning both the DAG $G_{\sigma}$ and the resulting covariance structure,   a natural attempt is to impose a prior on the DAG $G$   along with a prior on the covariance structure given $\sigma$ or the DAG,  and then perform posterior inference.   
	However, some identifiability conditions need to be imposed for the identifiability of the DAG structure. 
	Please see Theorem 2.2 of \cite{Park2019IdentifiabilityOG} for a state-of-the-art result  on the  identifiability of a Gaussian DAG, which essentially says that a Gaussian DAG is identifiable if the uncertainty level of a node $j $ is smaller than that of its descendants,  given the non-descendants.  
	One can potentially design a prior for the variances ($d_j$s in our notation) that satisfies this ordered constraint.  
	It would be interesting to study the posterior contraction  and model selection properties of this model and design efficient MCMC algorithms.

	If one does not impose any identifiability conditions (e.g., through the prior), another direction to deal with the unknown ordering case is to  learn the equivalent class (or a representative of the class) of a DAG for all $\sigma$s, the so-called \emph{structure learning of DAGs} which will learn the covariance structure of an equivalent class without learning the topological ordering  of the variable or the underlying DAG structure.  
	A recent work \cite{zhou2021complexity} is one such example.  The key ideas are the following: Let $N_p(0, \Sigma^*)$ be the distribution of some true Gaussian DAG  model $G^*$. 
	One can show that (see definition 7 in  \cite{zhou2021complexity} for example), for any $\sigma\in \mathcal S_p$, one can have the Cholesky factor matrix $A_\sigma^*$ and  the positive diagonal matrix $D_\sigma^*$ such that  $\Sigma^*  = (I_p - A_\sigma^*)^{-1} D_\sigma^* (I_p-  (A_\sigma^*)^T )^{-1}$.  
	That is, for any $\sigma$, one can find a unique pair ($A_\sigma^*$, $D_\sigma^*)$  which corresponds to some DAG $G_\sigma^*$ called minimal I-map of the equivalent classes which has the same covariance structure as $G^*$.  
	This minimal I-map can be viewed as a representative point of the equivalent class which can be uniquely constructed.  
	The structure learning problems boil down to the learnings of  $\{G_\sigma^*, \sigma\in \mathcal S_p\}$ essentially.   
	One can design appropriate priors  that can obtain posterior consistency over  $\{G_\sigma^*, \sigma\in \mathcal S_p\}$.

	Another work  that deals with unknown ordering is \cite{permutation} which proposed to sample $K$ permutation matrices or $\sigma$s under the MCD model.
	They obtained the posterior estimates of the Cholesky factors and diagonal matrices for each permutation which are then averaged to obtain a final estimate of the precision matrix.
	\ech


	\section*{Acknowledgement}
	We are very grateful to the Associate Editor and the two reviewers for their valuable comments  which have led to great improvement in our paper.  We  would also  like to thank Guo (Hugo) Yu for very helpful discussions.


	\section*{Appendix}
	\appendix

	\section{Proofs}\label{sec_proofs}

	\begin{proof}[Proof of Theorem \ref{thm_bandwidth_sel}]
		Note that 
		\bea
		&& \pi_\alpha \big(  k_j \neq k_{0j} \text{ for at least one } 2\le j \le p  \mid \bfX_n \big) \\
		&\le& \sum_{j=2}^p \pi_\alpha ( k_j \neq k_{0j} \mid \bfX_n)    \\
		&=& \sum_{j=2}^p \Big\{   \pi_\alpha ( k_j > k_{0j} \mid \bfX_n)  +  \pi_\alpha ( k_j < k_{0j} \mid \bfX_n)    \Big\}  .
		\eea
		We first show that 
		\bea
		\sum_{j=2}^p \bbE_0 \big\{   \pi_\alpha ( k_j > k_{0j} \mid \bfX_n) \big\}  &\lra& 0
		\eea
		as $n\to\infty$.
		For any $k_j > k_{0j}$, we have
		\bea
		\pi_\alpha(k_j  \mid \bfX_n) 
		&\le& \frac{\pi_\alpha(k_j  \mid \bfX_n) }{\pi_\alpha(k_{0j}  \mid \bfX_n) }    \\
		&\le& \frac{\pi(k_j)}{\pi(k_{0j})} \Big( 1+ \frac{\alpha}{\gamma}  \Big)^{-\frac{k_j - k_{0j}}{2}}  \Bigg(  \frac{\what{d}_{j}^{(k_j)} }{\what{d}_{j}^{(k_{0j})}}  \Bigg)^{-\frac{\alpha n + \nu_0}{2} }  .
		\eea
		Furthermore, it is easy to see that $\what{d}_{j}^{(k_{j})} / \what{d}_{j}^{(k_{0j})} \sim Beta ( (n-k_j)/2 , (k_j-k_{0j})/2 )$ for any $1\le j \le p$ because 
		\bea
		n d_{0j}^{-1}  \what{d}_{j}^{(k_{j})} &=& d_{0j}^{-1} \tilde{X}_j^T ( I_n - \tilde{P}_{j k_j}) \tilde{X}_j \sim \chi_{n -k_j}^2  , \\
		n d_{0j}^{-1}  \what{d}_{j}^{(k_{0j})} &\overset{d}{\equiv}& n d_{0j}^{-1}  \what{d}_{j}^{(k_{j})}  \oplus  \chi_{k_j-k_{0j}}^2   .
		\eea 
		Thus,
		\bea
		\bbE_0  \Bigg(  \frac{\what{d}_{j}^{(k_j)} }{\what{d}_{j}^{(k_{0j})}}  \Bigg)^{-\frac{\alpha n + \nu_0}{2} } 
		&=&  \frac{ \Gamma \Big( \frac{n-k_{0j}}{2} \Big)  \Gamma \Big(  \frac{n(1-\alpha) - \nu_0 - k_j }{2} \Big)  }{ \Gamma \Big( \frac{n-k_j}{2}  \Big)    \Gamma \Big( \frac{n(1-\alpha) - \nu_0 - k_{0j} }{2} \Big) }   \\
		&\le&  \Big(  \frac{2}{1-\alpha}  \Big)^{\frac{k_j - k_{0j}}{2} } ,
		\eea
		because we assume that $\nu_0 = O(1)$ and $R_j \le n(1-\alpha)/2$.
		This implies that 
		\bea
		\sum_{j=2}^p \bbE_0 \big\{   \pi_\alpha ( k_j > k_{0j} \mid \bfX_n) \big\}  
		&=& \sum_{j=2}^p  \sum_{k_j = k_{0j}+1}^{R_j} \bbE_0 \big\{  \pi_\alpha (k_j \mid\bfX_n)  \big\} \\
		&\le& \sum_{j=2}^p  \sum_{k_j = k_{0j}+1}^{R_j}   \Big(  \frac{c_{\alpha,\gamma} }{c_1 p^{c_2} } \Big)^{k_j - k_{0j}} \\
		&\lesssim& \sum_{j=2}^p  \frac{c_{\alpha,\gamma} }{c_1 p^{c_2}}  ,
		\eea
		where $c_{\alpha,\gamma} = (1+ \alpha/\gamma)^{-1/2} \{2/(1-\alpha)\}^{1/2}$. 
		The last term is of order $o(1)$ because we assume that $c_2 >1$.

		Now we show that 
		\bea
		\sum_{j=2}^p \bbE_0 \big\{   \pi_\alpha ( k_j < k_{0j} \mid \bfX_n) \big\}  &\lra& 0
		\eea
		as $n\to\infty$.
		For given integer $2\le j \le p$ and  constant $\epsilon = \{ (1 - \alpha) /10 \} ^2$, define the sets 
		\bea
		N_j^c &=& \Big\{  \bfX_n : \lambda_{\max}(\Omega_{0n})^{-1}  (1-2\epsilon)^2 \le n^{-1} \lambda_{\min}( \bfX_{j (k_{0j})}^T \bfX_{j (k_{0j})} ) \\
		&& \quad\quad\quad \quad\quad \le  n^{-1}\lambda_{\max} (\bfX_{j (k_{0j})}^T \bfX_{j (k_{0j})}) \le \lambda_{\min}  (\Omega_{0n})^{-1} (1+2\epsilon)^2   \Big\}  ,  \\
		N_{1j}^c &=& \Big\{  \bfX_n : \Big| \frac{\what{d}_j^{(k_{0j})} }{d_{0j} }   -1 \Big| \in  \Big(  -4\sqrt{\epsilon} \frac{n - k_{0j}}{n} - \frac{k_{0j}}{n}  , \, 4\sqrt{\epsilon} \frac{n - k_{0j}}{n} - \frac{k_{0j}}{n}     \Big)  \Big\}  , \\
		N_{2 j k_j}^c  &=&  \Big\{  \bfX_n:  0 <   \frac{\what{d}_j^{(k_j)} -  \what{d}_j^{(k_{0j})}  }{d_{0j}}  < \epsilon + \frac{\what{\lambda}_{j k_j} }{n}   \Big\} ,
		\eea
		where 
		\bea
		\what{\lambda}_{j k_j}  &=& \frac{1}{d_{0j}} \| (I_n - \tilde{P}_{j k_j} )  \bfX_{j (k_{0j})}  a_{0j}^{(k_{0j})}  \|_2^2  . 
		\eea
		Let $N_{j k_j}^c = N_j^c \cap N_{1j}^c \cap N_{2 j k_j}^c$ for all $2\le j \le p$.
		Since 
		\bea
		&& \sum_{j=2}^p \bbE_0 \big\{   \pi_\alpha ( k_j < k_{0j} \mid \bfX_n) \big\}  \\
		&\le&  \sum_{j=2}^p \sum_{k_j < k_{0j}} \bbP_0 (N_{j k_j })  +  \sum_{j=2}^p \sum_{k_j < k_{0j}} \bbE_0 \Big\{ \frac{\pi_\alpha(k_j \mid \bfX_n) }{\pi_\alpha(k_{0j} \mid \bfX_n) }    I(  N_{j k_j}^c) \Big\}    \\
		&\hspace{-.3cm} \le&\hspace{-.3cm}  \sum_{j=2}^p \sum_{k_j < k_{0j}} \bbP_0 (N_{j k_j })  +  \sum_{j=2}^p \sum_{k_j < k_{0j}} \bbE_0 \Bigg\{  \frac{\pi(k_j)}{\pi(k_{0j})} \Big( 1 + \frac{\alpha}{\gamma} \Big)^{k_{0j}- k_j }  \Bigg( \frac{\what{d}_j^{(k_j)} }{\what{d}_j^{(k_{0j})} }   \Bigg)^{- \frac{\alpha n + \nu_0}{2} }   I(  N_{j k_j}^c) \Bigg\}   ,
		\eea
		we will complete the proof by showing that 
		\bean\label{Null_sets}
		\sum_{j=2}^p \sum_{k_j < k_{0j}} \bbP_0 (N_{j k_j })  &\lra& 0
		\eean
		and
		\bean\label{small_kj}
		\sum_{j=2}^p \sum_{k_j < k_{0j}} \bbE_0 \Bigg\{  \frac{\pi(k_j)}{\pi(k_{0j})} \Big( 1 + \frac{\alpha}{\gamma} \Big)^{k_{0j}- k_j }  \Bigg( \frac{\what{d}_j^{(k_j)} }{\what{d}_j^{(k_{0j})} }   \Bigg)^{- \frac{\alpha n + \nu_0}{2} }   I(  N_{j k_j}^c) \Bigg\}  &\lra& 0
		\eean
		as $n\to\infty$.

		To show \eqref{Null_sets}, it suffices to show that 
		\bea
		\sum_{j=2}^p \sum_{k_j < k_{0j}} \Big\{   \bbP_0 (N_{j }) + \bbP_0 (N_{1j })+ \bbP_0 (N_{2j k_j })    \Big\}  &\lra& 0  
		\eea
		as $n\to\infty$.
		Note that 
		\bea
		\sum_{j=2}^p \sum_{k_j < k_{0j}}  \bbP_0 (N_{j }) 
		&\le& \sum_{j=2}^p \sum_{k_j < k_{0j}}  4 \exp \Big(  - \frac{n\epsilon^2 }{2} \Big)  \\
		&\le& \sum_{j=2}^p 4 \exp \Big(  - \frac{n\epsilon^2 }{2}  + \log k_{0j}    \Big) \\
		&\le&   4 \exp \Big(  - \frac{n\epsilon^2 }{2}  + 2\log p    \Big) \,\, =\,\, o(1)
		\eea
		for all sufficiently large $n$, where the first inequality follows from Corollary 5.35 in \cite{eldar2012compressed}.
		Since $n \what{d}_j^{(k_{0j})} / d_{0j} \sim \chi^2_{n - k_{0j}}$, we have 
		\bea
		\sum_{j=2}^p \sum_{k_j < k_{0j}}  \bbP_0 (N_{1j }) 
		&\le&  \sum_{j=2}^p \sum_{k_j < k_{0j}}  2 \exp  \{ - \epsilon ( n - k_{0j}) \}  \\
		&\le& \sum_{j=2}^p 2 \exp  \{ - \epsilon ( n - k_{0j})  + \log k_{0j}  \}   \\
		&\le&   2 \exp  \Big( - \frac{\epsilon n}{2}  + 2\log p  \Big)  \,\, =\,\, o(1)
		\eea
		for all sufficiently large $n$,  by the concentration inequality for chi-square random variables in Lemma 1 of \cite{laurent2000adaptive}.
		Furthermore, because $n(\what{d}_j^{(k_j)} -  \what{d}_j^{(k_{0j})} )/ d_{0j} \sim \chi^2_{k_{0j}-k_j}(\what{\lambda}_{j k_j})$, where $\chi_m^2(\lambda)$ denotes the noncentral chi-square distribution with the degrees of freedom $m$ and the noncentrality parameter $\lambda$, we have
		\bea
		&&  \sum_{j=2}^p \sum_{k_j < k_{0j}}  \bbP_0 (N_{2 k_j }) \\
		&\lesssim& \sum_{j=2}^p \sum_{k_j < k_{0j}}  \left[ \Big\{  \frac{\epsilon n}{2(k_{0j}- k_j )}  \Big\}^{\frac{k_{0j}-k_j}{2}} \exp \Big(\frac{k_{0j}-k_j  }{2} - \frac{\epsilon n}{2} \Big)  + \bbE_0 \Big( \frac{\what{\lambda}_{j k_j}}{\epsilon n} e^{-\frac{\epsilon^2 n^2}{32 \what{\lambda}_{j k_j} } }  \wedge  1 \Big)  \right]  \\ 
		&\le& \sum_{j=2}^p \sum_{k_j < k_{0j}} \left[  \exp \Big( - \frac{\epsilon n}{4} \Big) + \bbE_0 \Big\{ \frac{\what{\lambda}_{j k_j}}{\epsilon n} e^{-\frac{\epsilon^2 n^2}{32 \what{\lambda}_{j k_j} } }  I(N_j^c) \Big\}  + \bbP_0 (N_j)   \right]  \\
		&\le&  \sum_{j=2}^p \sum_{k_j < k_{0j}}   \left[  \exp \Big( - \frac{\epsilon n}{4} \Big) + \exp \Big\{ - \frac{\epsilon^2 n \lambda_{\max}(\Omega_{0n})^{-1}\lambda_{\min}(\Omega_{0n})   }{128(1+2\epsilon)^2}   \Big\} + 4 \exp \Big(  - \frac{n\epsilon^2 }{2} \Big)  \right]  \\
		&=& o(1)
		\eea
		by Lemma 4 of \cite{shin2018scalable} and conditions (A1), (A3) and (A4).
		The last inequality holds because 
		\bea
		\what{\lambda}_{j k_j} 
		&\le& \lambda_{\max} (\bfX_{j (k_{0j})}^T \bfX_{j (k_{0j})}) d_{0j}^{-1} \| a_{0j}^{(k_{0j})} \|_2^2  \\
		&\le& n \lambda_{\min}(\Omega_{0n})^{-1} (1+2\epsilon)^2 \big\{  d_{0j}^{-1} + \|d_{0j}^{-1/2}(e_j - a_{0j} )\|_2^2   \big\} \\
		&\le& 2n \lambda_{\min}(\Omega_{0n})^{-1} (1+2\epsilon)^2 \lambda_{\max}(\Omega_{0n}) ,
		\eea
		where $e_j$ is the unit vector whose the $j$th element is 1 and the others are zero.
		Thus, we have proved \eqref{Null_sets}.

		Now, we complete the proof by showing \eqref{small_kj}.
		For a given $k_j < k_{0j}$, because $-\log (1-x) \le x/(1-x)$ for any $x <1$, it can be shown that 
		\bea
		\Bigg( \frac{\what{d}_j^{(k_j)} }{\what{d}_j^{(k_{0j})} }   \Bigg)^{- \frac{\alpha n + \nu_0}{2} } 
		&=& \Bigg(  1 -  \frac{\what{d}_j^{(k_{0j})} - \what{d}_j^{(k_{j})} }{ \what{d}_j^{(k_{0j})} }   \Bigg)^{- \frac{\alpha n + \nu_0}{2} }   \\
		&\le&  \exp  \Bigg\{  \frac{ \alpha + \frac{\nu_0}{n} }{2d_{0j}( 1 + \what{Q}_{j k_j} ) }  n \Big( \what{d}_j^{(k_{0j})}  - \what{d}_j^{(k_{j})}   \Big)  \Bigg\}  ,
		\eea
		where 
		\bea
		\what{Q}_{j k_j}  &=&  \frac{\what{d}_j^{(k_{0j})} }{d_{0j} }   -1 +  \frac{\what{d}_j^{(k_j)} -  \what{d}_j^{(k_{0j})}  }{d_{0j} }   .
		\eea
		Then, we have
		\bea
		\what{Q}_{j k_j}  &\le& 4\sqrt{\epsilon} \frac{n - k_{0j}}{n} - \frac{k_{0j}}{n} + \epsilon + \frac{\what{\lambda}_{j k_j} }{n}  \,\,\le\,\, 5 \sqrt{\epsilon} +\frac{\what{\lambda}_{j k_j} }{n}   \,\,=:\,\, \what{Q}_{1j}  , \\
		\what{Q}_{j k_j}  &\ge& -4\sqrt{\epsilon} \frac{n - k_{0j}}{n} - \frac{k_{0j}}{n}  \,\,\ge\,\, -5 \sqrt{\epsilon} \,\,=:\,\, \what{Q}_{2j}
		\eea
		on the event $N_{j k_j}^c$.
		Further note that 
		\bea
		&&  n \big(  \what{d}_j^{(k_{0j})}  - \what{d}_j^{(k_j)}  \big) \\
		&=& \tilde{X}_j^T ( \tilde{P}_{j k_j}  - \tilde{P}_{j k_{0j}} ) \tilde{X}_j  \\
		&\overset{d}{\equiv}&  - \| (I_n  -\tilde{P}_{j k_j} )  \bfX_{j (k_{0j})} a_{0j}^{(k_{0j})} \|_2^2 - 2 \tilde{\epsilon}_j^T (I_n - \tilde{P}_{j k_j}) \bfX_{j (k_{0j})} a_{0j}^{(k_{0j})} + \tilde{\epsilon}_j^T  ( \tilde{P}_{j k_j}  - \tilde{P}_{j k_{0j}} ) \tilde{\epsilon}_j  \\
		&\le& - d_{0j} \what{\lambda}_{j k_j} - 2 \tilde{\epsilon}_j^T (I_n - \tilde{P}_{j k_j}) \bfX_{j (k_{0j})} a_{0j}^{(k_{0j})}  \\
		&=:& - d_{0j} \what{\lambda}_{j k_j}   -2 V_{j k_j} ,
		\eea
		where $\tilde{\epsilon}_j \sim N_n(0, d_{0j}I_n)$ and $V_{j k_j}/ \sqrt{d_{0j}} \sim N(0, d_{0j} \what{\lambda}_{j k_j} )$ under $\bbP_0$ given $\bfX_{j (k_{0j})}$.
		Then, 
		\bea
		&&  \bbE_0 \left\{  \Bigg( \frac{\what{d}_j^{(k_j)} }{\what{d}_j^{(k_{0j})} }   \Bigg)^{- \frac{\alpha n + \nu_0}{2} } I(N_{j k_j}^c)   \,\,\Big|\,\, \bfX_{j (k_{0j})} \right\}  \\
		&\le&  \bbE_0 \left[  \exp  \Bigg\{  \frac{ \alpha + \frac{\nu_0}{n} }{2d_{0j}( 1 + \what{Q}_{j k_j} ) }  n \Big( \what{d}_j^{(k_{0j})}  - \what{d}_j^{(k_{j})}   \Big)  \Bigg\}  I(N_{j k_j}^c)    \,\,\Big|\,\, \bfX_{j (k_{0j})} \right]     \\
		&\le&  \bbE_0 \left[  \sum_{ \what{Q}_{j k_j} \in \{ \what{Q}_{1j}, \what{Q}_{2j} \} }  \exp  \Bigg\{  -\frac{ \alpha + \frac{\nu_0}{n} }{2d_{0j}( 1 + \what{Q}_{j k_j} ) }  \Big( d_{0j} \what{\lambda}_{j k_j}   + 2 V_{j k_j}  \Big)  \Bigg\}  I(N_{j k_j}^c)   \,\,\Big|\,\, \bfX_{j (k_{0j})} \right]      \\
		&\le&  2\exp \Big\{ -\frac{ \alpha + \frac{\nu_0}{n} }{2( 1 +  5\sqrt{\epsilon} + \what{\lambda}_{j k_j}/n ) }  \Big(  1-    \frac{ \alpha + \frac{\nu_0}{n} }{  1 - 5\sqrt{\epsilon}  }  \Big)  \what{\lambda}_{j k_j}   \Big\}    \\
		&\le& 2\exp \Big\{ -\frac{ \alpha + \frac{\nu_0}{n} }{2( 1 +  5\sqrt{\epsilon} + \what{\lambda}_{j k_j}/n  ) }  \Big(  1-    \frac{ \alpha + \frac{\nu_0}{n} }{  1 - 5\sqrt{\epsilon}  }  \Big)  \what{\lambda}_{j k_j}   \Big\}    \,\, \equiv \,\, 2 \what{\Lambda}_{j k_j}   ,
		\eea
		where the third inequality follows from the moment generating function of normal distribution.
		Note that
		\bea
		1 -  \frac{ \alpha + \frac{\nu_0}{n} }{  1 - 5\sqrt{\epsilon}  } &>& \frac{1 - \alpha - \frac{\nu_0}{n} }{2}  ,
		\eea
		and by Lemma 5 of \cite{arias2014estimation}, 
		\bea
		\what{\lambda}_{j k_j}  &=& d_{0j}^{-1} \| (I_n  -\tilde{P}_{j k_j} )  \bfX_{j (k_{0j})} a_{0j}^{(k_{0j})} \|_2^2 \\
		&\ge&   \lambda_{\min} ( \bfX_{j (k_{0j})}^T \bfX_{j (k_{0j})} ) (k_{0j}-k_j) \min_{(j,l): a_{0,jl} \neq 0 }( a_{0,jl}^2 /d_{0j})     \\
		&\ge& n  \lambda_{\max}(\Omega_{0n})^{-1} (1-2\epsilon)^2 (k_{0j} - k_j)  \min_{(j,l): a_{0,jl} \neq 0 } (a_{0,jl}^2 /d_{0j} )  
		\eea
		on the event $N_{j k_j}^c$.
		Therefore, on the event $N_{j k_j}^c$,
		\bea
		&& \frac{ \alpha + \frac{\nu_0}{n} }{2( 1 +  5\sqrt{\epsilon} + \what{\lambda}_{j k_j}/n  ) }  \Big(  1-    \frac{ \alpha + \frac{\nu_0}{n} }{  1 - 5\sqrt{\epsilon}  }  \Big)  \what{\lambda}_{j k_j}  \\
		&\ge&  \frac{1}{4} \Big( \alpha + \frac{\nu_0}{n} \Big)\Big(  1- \alpha - \frac{\nu_0}{n} \Big)  \frac{\what{\lambda}_{j k_j} }{1+ 5\sqrt{\epsilon}+ \what{\lambda}_{j k_j} /n} \\
		&=& \frac{1}{4} \Big( \alpha + \frac{\nu_0}{n} \Big)\Big(  1- \alpha - \frac{\nu_0}{n} \Big)  \Big(  \frac{1+5\sqrt{\epsilon} }{\what{\lambda}_{j k_j}} + \frac{1}{n}  \Big)^{-1}\\
		&\ge& \frac{1}{4} \Big( \alpha + \frac{\nu_0}{n} \Big)\Big(  1- \alpha - \frac{\nu_0}{n} \Big)  \Big\{   \frac{(1+5\sqrt{\epsilon})  \lambda_{\max}(\Omega_{0n})    }{(1-2\epsilon)^2 (k_{0j}- k_j) \min_{(j,l): a_{0,jl} \neq 0 } (a_{0,jl}^2 /d_{0j} )  }  + \frac{1}{n}  \Big\}^{-1} \\
		&\ge& \frac{1}{4} \Big( \alpha + \frac{\nu_0}{n} \Big)\Big(  1- \alpha - \frac{\nu_0}{n} \Big)   \Big\{   \frac{(1+5\sqrt{\epsilon})    }{(1-2\epsilon)^2 (k_{0j}- k_j) C_{\rm bm}\log p  }  + \frac{1}{n}  \Big\}^{-1}  \\
		&\ge& \frac{1}{8} \Big( \alpha + \frac{\nu_0}{n} \Big)\Big(  1- \alpha - \frac{\nu_0}{n} \Big)   \frac{(1-2\epsilon)^2 (k_{0j}- k_j) C_{\rm bm} \log p  }{1 +5\sqrt{\epsilon} }\\
		&\ge& (k_{0j} - k_j) M_{\rm  bm} \log p 
		\eea
		for all sufficiently large $n$ and any $0.6 \le \alpha <1$, by conditions (A2) and (A4), where $C_{\rm bm} = 10 M_{\rm bm} / \{(\alpha + \nu_0/n )(1- \alpha - \nu_0/n) \}$.
		It implies that 
		\bea
		&&  \sum_{j=2}^p \sum_{k_j < k_{0j}} \bbE_0 \Big\{ \frac{\pi_\alpha(k_j \mid \bfX_n) }{\pi_\alpha(k_{0j} \mid \bfX_n) }    I(  N_{j k_j}^c) \Big\} \\
		&\le&  \sum_{j=2}^p \sum_{k_j < k_{0j}}  ( c_1 p^{c_2} )^{k_{0j} - k_j } \Big( 1+ \frac{\alpha}{\gamma}  \Big)^{\frac{k_{0j}-k_j}{2}} \bbE_0  \Big[ \what{\Lambda}_{j k_j} \{ I(  N_{j k_j}^c) + I(  N_{j k_j}) \}  \Big]  \\
		&\lesssim&  \sum_{j=2}^p \sum_{k_j < k_{0j}}  ( c_1 p^{c_2} )^{k_{0j} - k_j } \Big( 1+ \frac{\alpha}{\gamma}  \Big)^{\frac{k_{0j}-k_j}{2}} \Big[ \exp \Big\{ - (k_{0j}-k_j)M_{\rm bm} \log p  \Big\} + \bbP_0 (N_{j k_j})   \Big]  \\
		&\lesssim&  \sum_{j=2}^p \sum_{k_j < k_{0j}}  \bigg\{  \frac{(1+ \alpha / \gamma )^{1/2}}{ p^{M_{\rm bm} - c_2 } }   \bigg\}^{k_{0j}-k_j} \,\,=\,\, o(1)  ,
		\eea
		because we assume that $M_{\rm bm} > c_2 +1$.
		This completes the proof.
	\end{proof}

	\begin{proof}[Proof of Theorem \ref{thm_post_conv}]
		By Theorem \ref{thm_bandwidth_sel}, we can focus on the event  $\cap_{j=2}^p \{k_j : k_j = k_{0j}  \}$.
		Then, on the event $\cap_{j=2}^p \{k_j : k_j = k_{0j}  \}$, we have
		\bea
		\|A_n - A_{0n}\|_{\max} &=& \max_{2\le j \le p} \|a_j^{(k_{0j})} - a_{0j}^{(k_{0j})} \|_{\max} \\
		&\le& \max_{2\le j \le p} \|a_j^{(k_{0j})} - a_{0j}^{(k_{0j})} \|_2 ,   \\
		\|A_n - A_{0n}\|_{\infty} &=&  \max_{2\le j \le p} \|a_j^{(k_{0j})} - a_{0j}^{(k_{0j})} \|_{1}  \\
		&\le& \max_{2\le j \le p} \sqrt{k_{0}}\|a_j^{(k_{0j})} - a_{0j}^{(k_{0j})} \|_2  .
		\eea
		\bch 
		Furthermore, because $\|A_n - A_{0n}\|_{F}^2 =   \sum_{j=2}^p \|a_j^{(k_{0j})} - a_{0j}^{(k_{0j})} \|_2^2 $ and $(p-1) \log p \le 4 \sum_{j=2}^p \log j$, 
		\bea
		&&  \bbE_0  \Big\{ \pi_\alpha \Big(   \sum_{j=2}^p \|a_j^{(k_{0j})} - a_{0j}^{(k_{0j})} \|_2^2   \ge C \,  \frac{ \sum_{j = 2}^p(k_{0j} + \log j)}{n}   \mid \bfX_n \Big)  \Big\}  \\
		&\le&  \sum_{j=2}^p\bbE_0  \Big\{ \pi_\alpha \Big(    \|a_j^{(k_{0j})} - a_{0j}^{(k_{0j})} \|_2^2   \ge  \frac{C}{4} \,  \frac{ (k_{0j} + \log p )}{n}   \mid \bfX_n \Big)  \Big\}  
		\eea
		for any constant $C>0$.
		If we show that 
		\bean\label{aj_vecL2_conv}
		\sum_{ j = 2}^p \bbE_0  \Big\{ \pi_\alpha \Big(    \|a_j^{(k_{0j})} - a_{0j}^{(k_{0j})} \|_2^2 \ge K \,\frac{\lambda_{\max}(\Omega_{0n})^4 }{\lambda_{\min}(\Omega_{0n})^4 } \frac{k_{0j} + \log p}{n}   \mid \bfX_n \Big)  \Big\}  &\lra& 0 
		\eean
		as $n\to\infty$, for some constant $K>0$, it completes the proof.
		
		To this end, we need to carefully modify the proofs of Lemmas 6.1--6.5 in \cite{lee2019minimax}.
		Define the sets	
		\bea
		\tilde{N}_j^c &=& \Big\{  \bfX_n : \lambda_{\max}(\Omega_{0n})^{-1}  (1-2\epsilon)^2 \le \min_{k_j \le R_j} n^{-1} \lambda_{\min}( \bfX_{j (k_{j})}^T \bfX_{j (k_{j})} ) \\
		&& \quad\quad\quad \quad\quad \le  \max_{k_j \le R_j} n^{-1}\lambda_{\max} (\bfX_{j (k_{j})}^T \bfX_{j (k_{j})}) \le \lambda_{\min}  (\Omega_{0n})^{-1} (1+2\epsilon)^2   \Big\}  , \\
		\tilde{N}_{3, k_{0}} &=& \bigcup_{2\le j \le p}  \Big\{  \bfX_n : \| \what{\V}( \bfX_{j(k_j +)} )  - \V (X_{j(k_j +)} )  \|   \ge  \frac{C}{\lambda_{\min}(\Omega_{0n}) }\sqrt{\frac{k_{0j} + \log p}{n}} \,\, \Big\}  \\
		\tilde{N}_{4, k_{0}} &=&  \bigcup_{2\le j \le p}  \Big\{  \bfX_n : \| \what{\V}^{-1}( \bfX_{j(k_j +)} )  - \V^{-1} (X_{j(k_j +)} )  \|   \\
		&& \quad\quad\quad\quad\quad\quad\quad\quad\quad\quad\quad\quad \ge \,\, \frac{ C \lambda_{\max}(\Omega_{0n})^2 }{(1-2\epsilon)^2 \lambda_{\min}(\Omega_{0n}) }\sqrt{\frac{k_{0j} + \log p}{n}} \,\, \Big\} 
		\eea
		for some constant $C>0$, 
		where $\bfX_{j(k_j +)}  = (\bfX_{j(k_j)} , \tilde{X}_j) \in \bbR^{n \times (k_j +1)}$ and $\what{\V}( \bfX_{j(k_j +)} )  = n^{-1} \bfX_{j(k_j +)}^T \bfX_{j(k_j +)}$.
		Let $\tilde{N}_{\rm union} = \big(\cup_{2\le j \le p} \tilde{N}_j\big) \cup \tilde{N}_{3, k_{0}} \cup \tilde{N}_{4, k_{0}}$.
		By following closely the line of the proofs of Lemma 6.1 and 6.2 in \cite{lee2019minimax} and  
		\bea
		\|  \what{\V}^{-1}( \bfX_{j(k_j +)} )  - \V^{-1} (X_{j(k_j +)} ) \| 
		&\le& \| \what{\V}^{-1}( \bfX_{j(k_j +)} )\|   \| \V^{-1} (X_{j(k_j +)} )\| \\
		&& \times \,\, \| \what{\V}( \bfX_{j(k_j +)} )  - \V (X_{j(k_j +)} )  \|   \\
		&\le&  \frac{\lambda_{\max}(\Omega_{0n})^2 }{(1- 2\epsilon)^2 } \| \what{\V}( \bfX_{j(k_j +)} )  - \V (X_{j(k_j +)} )  \|  
		\eea
		on $\tilde{N}_j^c$,
		we have 
		\bea
		\bbP_0 \big( \tilde{N}_{\rm union}    \big)  
		&\lesssim&  \sum_{j=2}^p \sum_{k_j \le R_j} \exp   \Big(  - \frac{n}{2} \epsilon^2  \Big) 
		+ \sum_{j=2}^p 5^{k_{0j}} \exp \Big\{  - C   (k_{0j}+ \log p )  \Big\}   \\
		&\lesssim&  \exp   \Big(  - \frac{n}{2} \epsilon^2  + 2\log (n\vee p)  \Big) 
		+ \max_{2\le j \le p} \exp \Big\{  -  \frac{C}{2}  (k_{0j}+ \log p )  \Big\}  \\
		&=& o (1) 
		\eea
		for some large constant $C>0$.
		Thus, in the rest, we focus on the event $\tilde{N}_{\rm union}^c$.
		Then,
		\bean
		\| \what{a}_j^{(k_{0j})}  -   a_{0j}^{(k_{0j})} \|_2  
		&\le& \| \V^{-1} (\bfX_{j(k_{0j})}) \{ \what{\C}(\bfX_{j(k_{0j})} , \tilde{X}_j)  - \what{\C}(X_{j(k_{0j})} , {X}_j) \} \|_2  \nonumber\\
		&+&  \|  \{ \what{\V}^{-1} (\bfX_{j(k_{0j})}) - \V^{-1} (X_{j(k_{0j})}) \}  \what{\C} (\bfX_{j(k_{0j})} , \tilde{X}_j)  \|_2  \nonumber \\
		&\le& \| \V^{-1} (\bfX_{j(k_{0j})}) \|  \| \what{\V}( \bfX_{j(k_j +)} )  - \V (X_{j(k_j +)} )  \|  \nonumber \\
		&+& \| \what{\V}^{-1}( \bfX_{j(k_j +)} )  - \V^{-1} (X_{j(k_j +)} )  \|  \|\what{\V}(\bfX_{j(k_j +)}) \|  \nonumber \\
		&\lesssim& \frac{\lambda_{\max}(\Omega_{0n})^2 }{\lambda_{\min}(\Omega_{0n})^2 } \sqrt{ \frac{k_{0j} + \log p }{n} } . \label{aj_vecL2_conv1}
		\eean
		
		On the other hand, 
		\bea
		\| a_j^{(k_{0j})} - \what{a}_j^{(k_{0j})} \|_2 
		&\lesssim& \lambda_{\max}(\Omega_{0n})  \sqrt{\frac{d_j}{n}} \Big\|  \sqrt{\frac{n(\alpha+\gamma)}{d_j}}\what{\V}^{1/2}(\bfX_{j (k_{0j})})  \big( a_j^{(k_{0j})} - \what{a}_j^{(k_{0j})} \big)    \Big\|_2   \\
		&=:& \lambda_{\max}(\Omega_{0n}) \sqrt{\frac{d_j}{n}}  \big\|  std( a_j^{(k_{0j})} )  \big\|_2
		\eea
		and
		\bea
		&&  \pi_\alpha \Big(  \| a_j^{(k_{0j})} - \what{a}_j^{(k_{0j})} \|_2 \ge  \frac{C \lambda_{\max}(\Omega_{0n}) }{ \sqrt{\lambda_{\min}(\Omega_{0n})} }   \sqrt{\frac{ k_{0j} + \log p }{n}} \mid \bfX_n\Big) \\
		&\le& \pi_\alpha \Big(  \sqrt{d_j} \big\|  std( a_j^{(k_{0j})} )  \big\|_2 \ge   \frac{C'}{  \sqrt{\lambda_{\min}(\Omega_{0n})}} \sqrt{ k_{0j} + \log p } ,\, d_j \le \frac{ 2(1+2\epsilon)^2}{\lambda_{\min}(\Omega_{0n}) }\mid \bfX_n\Big)  \\
		&+& \pi_\alpha \Big( d_j > 2(1+2\epsilon)^2 \lambda_{\min}(\Omega_{0n})^{-1} \mid \bfX_n \Big)  \\
		&\le& \pi_\alpha \Big(  \big\|  std( a_j^{(k_{0j})} )  \big\|_2 \ge   C'' \sqrt{ k_{0j} + \log p } \mid \bfX_n\Big)  \\
		&+& \pi_\alpha \Big( d_j > 2(1+2\epsilon)^2 \lambda_{\min}(\Omega_{0n})^{-1} \mid \bfX_n \Big)
		\eea
		for some positive constants $C, C'$ and $C''$.
		By the proof of Lemma 6.5 in \cite{lee2019minimax} and page 29 of \cite{boucheron2013concentration}, 
		\bean
		\exp \Big\{ - \alpha n \Big( \frac{\epsilon}{2}\Big)^2  \Big\} 
		&\ge&  \pi_\alpha \Big(   d_j^{-1} - \frac{\alpha n +\nu_0}{\alpha n \what{d}_j^{(k_j)}   }  <  - \frac{1}{\what{d}_j^{(k_j)}  } \Big( \epsilon \sqrt{\frac{\alpha n + \nu_0}{\alpha n }} + \frac{\epsilon^2}{2}   \Big) \mid \bfX_n \Big)  \nonumber\\
		&\ge& \pi_\alpha \Big(   d_j^{-1} - \frac{\alpha n +\nu_0}{\alpha n \what{d}_j^{(k_j)}   }  <   \frac{\lambda_{\min}(\Omega_{0n}) }{2(1+\epsilon)^2 } - \frac{\alpha n +\nu_0}{\alpha n \what{d}_j^{(k_j)}  }     \mid \bfX_n \Big)  \nonumber    \\
		&\ge& \pi_\alpha \Big( d_j > 2(1+2\epsilon)^2 \lambda_{\min}(\Omega_{0n})^{-1} \mid \bfX_n \Big) . \label{ineq_M2}
		\eean
		Because $\|std( a_j^{(k_{0j})} )  \|_2^2 \mid \bfX_n \sim \chi_{k_{0j}}^2$ and Lemma 1 in \cite{laurent2000adaptive}, 
		\bea
		\pi_\alpha \Big(  \big\|  std( a_j^{(k_{0j})} )  \big\|_2 \ge   C'' \sqrt{ k_{0j} + \log p } \mid \bfX_n\Big)   
		&=&  o(1)  
		\eea
		for some constant $C''>0$, 
		which implies
		\bean\label{aj_vecL2_conv2}
		\pi_\alpha \Big(  \| a_j^{(k_{0j})} - \what{a}_j^{(k_{0j})} \|_2 \ge   \frac{C \lambda_{\max}(\Omega_{0n}) }{\lambda_{\min}(\Omega_{0n}) }   \sqrt{\frac{ k_{0j} + \log p }{n}} \mid \bfX_n\Big)  &=& o(1) 
		\eean
		for some constant $C>0$.
		Then, \eqref{aj_vecL2_conv1} and \eqref{aj_vecL2_conv2} imply that \eqref{aj_vecL2_conv} holds with some constant $K>0$ not depending on unknown parameters.
		\ech

		%
		
	\end{proof}

	\begin{proof}[Proof of Theorem \ref{thm_post_conv_nobetamin}]\text{}\\
		{\bf The posterior convergence rate under the matrix $\ell_\infty$-norm. } 
		For some constant $K_{\rm chol} >0$, let $\delta_n = K_{\rm chol} \sqrt{ \lambda_{\max}(\Omega_{0n}) /\lambda_{\min}(\Omega_{0n}) } \sqrt{k_0 \log p /n}$. 
		Note that
		\bean
		&& \bbE_0 \big\{   \pi_\alpha \big(   \|A_n - A_{0n} \|_{\infty}  \ge \sqrt{k_0} \delta_n   \mid \bfX_n \big)  \big\}  \nonumber \\
		&\le& \sum_{j=2}^p \bbE_0 \big\{   \pi_\alpha \big(   \|a_j - a_{0j} \|_{1}  \ge \sqrt{k_0} \delta_n   \mid \bfX_n \big)  \big\}   \nonumber \\
		&\le& \sum_{j=2}^p \bbE_0 \big\{   \pi_\alpha \big(   \|a_j - a_{0j} \|_{1}  \ge \sqrt{k_0} \delta_n , \, k_j \le C_{\rm dim} k_0  \mid \bfX_n \big) \big\}  \label{aj_con} \\
		&&+ \,\, \sum_{j=2}^p \bbE_0  \big\{   \pi_\alpha \big(   k_j  > C_{\rm dim} k_0  \mid \bfX_n \big) \big\}  \label{k_large}
		\eean
		for some constant $C_{\rm dim} >0$.
		
		We first focus on \eqref{k_large}. 
		We will show that 
		\bean
		&& \sum_{j=2}^p \bbE_0  \big\{   \pi_\alpha \big(   k_j  > C_{\rm dim} k_0  \mid \bfX_n \big) \big\}  \nonumber \\
		&\le& \sum_{j=2}^p \bbE_0  \big\{   \pi_\alpha \big(   k_j  > C_{\rm dim} k_0 , \,\, M_1 \le d_j \le M_2 \mid \bfX_n \big) \big\}   + o(1)  \label{dj_supp}
		\eean
		for  $M_1 = (1-2\epsilon)^2 \lambda_{\max}(\Omega_{0n})^{-1}  /2$ and $M_2 =  2(1+2\epsilon)^2 \lambda_{\min}(\Omega_{0n})^{-1}$.
		Let $\tilde{N}_{\rm union}$ be the set defined at the proof of Theorem \ref{thm_post_conv}.
		Note that by similar arguments used in \eqref{ineq_M2}, the proof of Lemma 6.5 in \cite{lee2019minimax} and page 29 of \cite{boucheron2013concentration},
		\bea
		\pi_\alpha ( M_1 \le d_j \le M_2 \mid  \bfX_n ) \ge 1 - 2 e^{- C n}
		\eea
		for some constant $C>0$ and any integer $2 \le j \le p$ on the event $\tilde{N}_{\rm union}^c$.
		Therefore, we have
		\bea
		&& \sum_{j=2}^p \bbE_0  \big\{   \pi_\alpha \big(   k_j  > C_{\rm dim} k_0  \mid \bfX_n \big) \big\}  \\
		&\le& \sum_{j=2}^p \bbE_0  \big\{   \pi_\alpha \big(   k_j  > C_{\rm dim} k_0 , \, M_1\le d_j \le M_2 \mid \bfX_n \big) \big\}  \\
		&&+\,\, \sum_{j=2}^p   \bbE_0  \big\{   \pi_\alpha \big(   d_j \in [M_1, M_2]^c \mid \bfX_n \big) I( \tilde{N}_{\rm union}^c ) \big\}  
		+ \bbP_0 \big(  \tilde{N}_{\rm union} \big)   \\
		&\le& \sum_{j=2}^p \bbE_0  \big\{   \pi_\alpha \big(   k_j  > C_{\rm dim} k_0 , \, M_1\le d_j \le M_2 \mid \bfX_n \big) \big\}  + o(1),
		\eea
		which implies that \eqref{dj_supp} holds.
		
		Now, we can show that \eqref{k_large} is of order $o(1)$ if we show that 
		\bean\label{kdj}
		\sum_{j=2}^p \bbE_0  \big\{   \pi_\alpha \big(   k_j  > C_{\rm dim} k_0 , \, M_1\le d_j \le M_2 \mid \bfX_n \big) \big\}  &=& o(1) . 
		\eean
		Let 
		\bea
		N_{nj}(k_j > C_{\rm dim} k_0 ) 
		&\hspace{-.3cm} =\hspace{-.3cm} &\hspace{-.3cm}  \sum_{k_j > C_{\rm dim} k_0} \int_{M_1}^{M_2} \int  R_{nj}( a_{j,+}^{(k_j)}, d_j)^\alpha \pi(a_j^{(k_j)} \mid d_j, k_j ) \pi(k_j) \pi(d_j)   d a_j^{(k_j)} d d_j     ,   \\
		R_{nj} (a_j , d_j)   &=& L_{nj}( a_j, d_j) / L_{nj}( a_{0j}, d_{0j})     ,
		\eea
		where
		\bea
		L_{nj}(a_j , d_j)  &=& (2\pi d_j)^{-n/2} \exp \big\{ - \| \tilde{X}_j - \bfX_{j(j-1)}a_j \|_2^2/ (2d_j)   \big\}
		\eea
		and $a_{j,+}^{(k_j)} $ be a $p$-dimensional vector such that $(a_{j,+}^{(k_j)})_{(j - k_j):(j-1)} = a_j^{(k_j)}$ and $(a_{j,+}^{(k_j)})_l =0$ otherwise.
		For any integer $2\le j \le p$, 
		\bea
		&& \bbE_0  \big\{   \pi_\alpha \big(   k_j  > C_{\rm dim} k_0 , \, M_1\le d_j \le M_2 \mid \bfX_n \big) \big\} \\
		&\le&   \bbE_0 \big\{  N_{nj}(k_j > C_{\rm dim} k_0 )  \big\} \frac{e^{\tilde{C}_2 k_{0j} } }{\pi(k_{0j}) } \tilde{C}_3 n^2 j^2  \\
		&&+\,\, \bbP_0 \Big(  \big| \|\tilde{X}_j\|_2^2 - \| \bfX_{j (k_{0j})} \what{a}_j^{(k_{0j})} \|_2^2  - n d_{0j}    \big| \le   \frac{1}{j^2 n}   \Big),
		\eea
		for some positive constants $\tilde{C}_2$ and $\tilde{C}_3$ depending only on $(\alpha, \gamma, \nu_0, \nu_0')$, by the proofs of Lemmas 6.5 and 7.1 in \cite{lee2019minimax}.
		Since $d_{0j}^{-1} (  \|\tilde{X}_j\|_2^2 - \| \bfX_{j (k_{0j})} \what{a}_j^{(k_{0j})} \|_2^2 ) \sim \chi_{n -k_{0j}}^2$ under $\bbP_0$, we have 
		\bean\label{chisq_con}
		\sum_{j=2}^p \bbP_0 \Big(  \big| \|\tilde{X}_j\|_2^2 - \| \bfX_{j (k_{0j})} \what{a}_j^{(k_{0j})} \|_2^2  - n d_{0j}    \big| \le   \frac{1}{j^2 n}   \Big)
		&\le& \sum_{j=2}^p \frac{1}{d_{0j} j^2  n } \,\,=\,\, o(1) 
		\eean
		because $d_{0j}^{-1} \le \lambda_{\max}(\Omega_{0n}) $ and $\lambda_{\max}(\Omega_{0n})  = o( n )$.
		By the proof of Lemma 7.2 in \cite{lee2019minimax}, 
		\bea
		&& \sum_{j=2}^p \bbE_0 \big\{  N_{nj}(k_j > C_{\rm dim} k_0 )  \big\} \frac{e^{\tilde{C}_2 k_{0j} } }{\pi(k_{0j}) } \tilde{C}_3 n^2 j^2   \\
		&\le& \sum_{j=2}^p \sum_{k_j > C_{\rm dim} k_0 } \pi(k_j)  \Big(  1 + C' \frac{ d_{0j} }{M_1}  \Big)^{k_j} \frac{e^{\tilde{C}_2 k_{0j} } }{\pi(k_{0j}) } \tilde{C}_3 n^2 j^2   \\
		&\le& \sum_{j=2}^p  \frac{e^{\tilde{C}_2 k_{0j} } }{\pi(k_{0j}) } \tilde{C}_3 n^2 j^2  \sum_{k_j > C_{\rm dim} k_0 } \pi(k_j)  \Big(  1 + C' \frac{ \lambda_{\max}(\Omega_{0n}) }{\lambda_{\min}(\Omega_{0n})}  \Big)^{k_j}  \\
		&\lesssim& \sum_{j=2}^p  \frac{e^{\tilde{C}_2 k_{0j} } }{\pi(k_{0j}) } \tilde{C}_3 n^2 j^2  \Big\{   1 + C' \frac{ \lambda_{\max}(\Omega_{0n}) }{\lambda_{\min}(\Omega_{0n})}  \Big\}^{ C_{\rm dim} k_0 } \Big( \frac{1}{c_1 p^{c_2} } \Big)^{C_{\rm dim} k_0}  \\
		&\lesssim& \sum_{j=2}^p \Big( \frac{ p }{c_1 p^{c_2} }  \Big)^{C_{\rm dim} k_0} \exp \big\{  \tilde{C}_2 k_{0j}  + (c_2+1) k_{0j}\log p +4 \log (n\vee j)    \big\}
		\,\, =\,\, o(1)
		\eea
		for some positive constant $C'$ and large constant $C_{\rm dim}>0$ depending only on $c_2$.
		Note that the fourth inequality holds because $\lambda_{\max}(\Omega_{0n})/ \lambda_{\min}(\Omega_{0n}) = O(p)$.
		Thus, it implies that \eqref{k_large} is of order $o(1)$.

		We can complete the proof if we show that \eqref{aj_con} is of order $o(1)$.
		Note that
		\bean
		&&  \sum_{j=2}^p \bbE_0 \big\{   \pi_\alpha \big(   \|a_j - a_{0j} \|_{1}  \ge \sqrt{k_0} \delta_n , \, k_j \le C_{\rm dim} k_0  \mid \bfX_n \big) \big\}  \nonumber\\
		&\le& \sum_{j=2}^p \bbE_0 \big\{   \pi_\alpha \big(   \|a_j - a_{0j} \|_{2}  \ge  (C_{\rm dim} +1)^{-1} \delta_n , \, k_j \le C_{\rm dim} k_0  \mid \bfX_n \big) \big\}   \label{aj_L2_kj}\\
		&\le& \sum_{j=2}^p \bbE_0 \Big\{   \pi_\alpha \Big(   \|a_j - a_{0j} \|_{2}  \ge  \frac{ (1-2\epsilon)  K_{\rm chol} }{(C_{\rm dim} +1)  \sqrt{\lambda_{\min}(\Omega_{0n})} }  \Big( \frac{k_0 \log p}{\Psi_{\min,j}(C_{\rm dim} k_0 ) } \Big)^{1/2}  ,   \nonumber\\
		&& \hspace{2.5cm}  k_j \le C_{\rm dim} k_0  \mid \bfX_n \Big) \Big\}  \,\, + \,\,  \bbP_0 \Big( \tilde{N}_{\rm union} \Big)   ,  \nonumber
		\eean
		where 
		\bea
		\Psi_{\min,j}( K ) &=&  \inf_{0 < k_j \le K}  \frac{1}{n} \lambda_{\min}( \bfX_{j (k_j)}^T \bfX_{j (k_j)} )  .
		\eea		
		Recall that  $\bbP_0 \big( \tilde{N}_{\rm union} \big) = o(1)$.
		Let $\tilde{C}_1 = (C_{\rm dim} +1)^{-1} (1-2\epsilon)$ and $\delta_{nj} = \sqrt{k_0 \log p /\Psi_{\min,j}(C_{\rm dim} k_0 )}$.
		Then, by the same arguments used in \eqref{dj_supp}, we have 
		\bea
		&& \sum_{j=2}^p \bbE_0 \Big\{   \pi_\alpha \Big(   \|a_j - a_{0j} \|_{2}  \ge  \frac{\tilde{C}_1 K_{\rm chol}}{  \sqrt{\lambda_{\min}(\Omega_{0n})} } \delta_{nj} , \, k_j \le C_{\rm dim} k_0  \mid \bfX_n \Big) \Big\}  \\
		&\le& \sum_{j=2}^p \bbE_0 \Big\{   \pi_\alpha \Big(   \|a_j - a_{0j} \|_{2}  \ge  \frac{\tilde{C}_1 K_{\rm chol}}{\sqrt{\lambda_{\min}(\Omega_{0n})} } \delta_{nj}  , \, k_j \le C_{\rm dim} k_0  ,\, M_1 \le d_j \le M_2  \mid \bfX_n \Big) \Big\}   + o(1).
		\eea
		By Lemma 7.1 in \cite{lee2019minimax}, 
		\bea
		&& \pi_\alpha \Big(   \|a_j - a_{0j} \|_{2}  \ge  \frac{\tilde{C}_1 K_{\rm chol}}{\sqrt{\lambda_{\min}(\Omega_{0n})} } \delta_{nj} , \, k_j \le C_{\rm dim} k_0  ,\, M_1 \le d_j \le M_2  \mid \bfX_n \Big)   \\
		&=&  \frac{ \sum_{0\le k_j \le C_{\rm dim} k_0} \int_{M_1}^{M_2} \int_{\|a_j - a_{0j} \|_{2}  \ge  \tilde{C}_1 K_{\rm chol} \delta_{nj}/\sqrt{\lambda_{\min}(\Omega_{0n})} } R_{nj}(a_{j,+}^{(k_j)}, d_j)^\alpha \pi(a_j^{(k_j)} \mid d_j, k_j ) \pi(k_j) \pi(d_j) d a_j^{(k_j)} d d_j   }{ \sum_{0\le k_j \le R_j} \int \int  R_{nj}(a_{j,+}^{(k_j)}, d_j)^\alpha \pi(a_j^{(k_j)} \mid d_j, k_j ) \pi(k_j) \pi(d_j) d a_j^{(k_j)} d d_j   }  \\
		&=:& \frac{N_{nj}}{D_{nj}}  \\
		&\le& N_{nj}  \frac{e^{\tilde{C}_2 k_{0j} } }{\pi(k_{0j})} \tilde{C}_3 n^2 j^2   + I \Big( \big| \|\tilde{X}_j\|_2^2 - \|\bfX_{j (k_{0j})} \what{a}_j^{(k_{0j})} \|_2^2 - n d_{0j}   \big| \le \frac{1}{j^2  n} \Big)
		\eea
		for some positive constants $\tilde{C}_2$ and $\tilde{C}_3$ depending only on $(\alpha, \gamma, \nu_0, \nu_0')$.
		Note that 
		\bea
		\bbP_0 \Big( \big| \|\tilde{X}_j\|_2^2 - \|\bfX_{j (k_{0j})} \what{a}_j^{(k_{0j})} \|_2^2 - n d_{0j}   \big| \le \frac{1}{j^2 \log n} \Big) &\le& \frac{1}{d_{0j} j^2  n}
		\eea
		by \eqref{chisq_con}, and 
		\bea
		\bbE_0 (N_{nj}) 
		&\le& \exp \Big( - \frac{1}{d_{0j} + M_2 } \frac{\tilde{C}_1^2 K_{\rm chol}^2}{\lambda_{\min}(\Omega_{0n}) } k_0 \log p     \Big)  \sum_{0\le k_j \le C_{\rm dim} k_0} \Big(  1 + C' \frac{d_{0j}}{M_1} \Big)^{k_j }  \pi (k_j)\\
		&\lesssim& \exp \Big( - C  \tilde{C}_1^2 K_{\rm chol}^2 k_0 \log p     \Big)  \sum_{0\le k_j \le C_{\rm dim} k_0} \Big\{   1 + C' \frac{ \lambda_{\max}(\Omega_{0n}) }{\lambda_{\min}(\Omega_{0n})}  \Big\}^{k_j }  \pi (k_j) \\
		&\lesssim& \exp \Big( - C  \tilde{C}_1^2 K_{\rm chol}^2 k_0 \log p     \Big)   \Big\{   1 + C' \frac{ \lambda_{\max}(\Omega_{0n}) }{\lambda_{\min}(\Omega_{0n})}  \Big\}^{C_{\rm dim} k_0 } 
		\eea
		for some constants $C, C'>0$, 
		by the proof of Lemma 7.3 of \cite{lee2019minimax}.
		Thus,
		\bea
		&& \bbE_0 (N_{nj})  \frac{e^{\tilde{C}_2 k_{0j} } }{\pi(k_{0j})} \tilde{C}_3 n^2 j^2     \\
		&\le& e^{- C  \tilde{C}_1^2 K_{\rm chol}^2 k_0 \log p }  \Big\{   1 + C' \frac{ \lambda_{\max}(\Omega_{0n}) }{\lambda_{\min}(\Omega_{0n})}  \Big\}^{C_{\rm dim} k_0 }       \frac{e^{\tilde{C}_2 k_{0j} } }{\pi(k_{0j})} \tilde{C}_3 n^2 j^2   \\
		&\le& \exp \big\{ - C  \tilde{C}_1^2 K_{\rm chol}^2 k_0 \log p +  C_{\rm dim}k_0 \log  p  + \tilde{C}_2 k_{0j} + 4 \log (n\vee j)  +  k_{0j} \log c_1 + c_2 k_{0j} \log p  \big\} .
		\eea
		Thus, it implies that 
		\bea
		\sum_{j=2}^p \bbE_0 \big\{   \pi_\alpha \big(   \|a_j - a_{0j} \|_{1}  \ge \sqrt{k_0} \delta_n , \, k_j \le C_{\rm dim} k_0  \mid \bfX_n \big) \big\} 
		&=&  o(1)
		\eea
		for some large constant $K_{\rm chol}>0$ depending only on $c_2$.

		\noindent
		{\bf The posterior convergence rate under the element-wise maximum norm. } 
		Note that 
		\bean
		&& \bbE_0 \big\{   \pi_\alpha \big(   \|A_n - A_{0n} \|_{\max}  \ge  \delta_n   \mid \bfX_n \big)  \big\}  \label{An_maximum}  \\
		&\le& \bbE_0 \big\{   \pi_\alpha \big(  \max_{2\le j \le p}  \| a_j - a_{0j} \|_{2}  \ge  \delta_n   \mid \bfX_n \big)  \big\}  \nonumber \\
		&\le& \sum_{2\le j \le p}  \bbE_0 \big\{   \pi_\alpha \big(   \| a_j - a_{0j} \|_{2}  \ge  \delta_n   \mid \bfX_n \big)  \big\} \nonumber \\
		&\le& \sum_{2\le j \le p}  \bbE_0 \big\{   \pi_\alpha \big(   \| a_j - a_{0j} \|_{2}  \ge  \delta_n , \, k_j \le C_{\rm dim} k_0   \mid \bfX_n \big)  \big\} \nonumber \\
		&& \,\, +\,\,   \sum_{2\le j \le p}  \bbE_0 \big\{   \pi_\alpha \big(   k_j  > C_{\rm dim} k_0   \mid \bfX_n \big)  \big\}  . \nonumber
		\eean
		Because we showed that \eqref{k_large} and \eqref{aj_L2_kj} are of order $o(1)$, it implies that \eqref{An_maximum} is of order $o(1)$.

		\noindent
		{\bf The posterior convergence rate under the Frobenius norm. } 
		Because $(p-1) \log p \le 4 \sum_{j=2}^p \log j$, we have 
		\bea
		&& \bbE_0 \Big\{   \pi_\alpha \big(   \|A_n - A_{0n} \|_{F}^2  \ge  K \frac{\lambda_{\max}(\Omega_{0n}) }{\lambda_{\min}(\Omega_{0n})} \frac{  k_0 \sum_{j=2}^p \log j}{n}   \mid \bfX_n \big)  \Big\}  \\
		&=& \bbE_0 \Big\{   \pi_\alpha \big(  \sum_{j=2}^p \|a_j - a_{0j} \|_{2}^2  \ge  K \frac{\lambda_{\max}(\Omega_{0n}) }{\lambda_{\min}(\Omega_{0n})} \frac{  k_0 \sum_{j=2}^p \log j}{n}   \mid \bfX_n \big)  \Big\}  \\
		&\le& \sum_{j=2}^p  \bbE_0 \Big\{   \pi_\alpha \big(  \|a_j - a_{0j} \|_{2}^2  \ge  K \frac{\lambda_{\max}(\Omega_{0n}) }{\lambda_{\min}(\Omega_{0n})}  \frac{  k_0 \log p}{4n}   \mid \bfX_n \big)  \Big\} .
		\eea
		The last term is of order $o(1)$ by the similar arguments used in \eqref{An_maximum}.
	\end{proof}

	\section{Numerical Studies: Estimation of Cholesky Factors}\label{sec:sim_est}
	
	In this section, we illustrate the performance of LANCE prior in terms of the estimation of Cholesky factors.
	We consider Models 1--3 described in Section \ref{subsec:simulated} of the main manuscript and compare the performance of LANCE prior with that of YB method \citep{yu2017learning}.
	The sample size and number of variables are varied as follows: $n \in \{100, 300\}$ and 
	$p\in \{100, 200\}$.
	To demonstrate the estimation performance of each method, cross-validation is used as described in Section \ref{subsec:simulated}.
	For LANCE prior, the posterior means of Cholesky factors $A_n$ given the posterior mode of bandwidths $k_j$ are used for estimation purpose.

	Tables \ref{table:est1} and \ref{table:est2} show the results under various loss functions when the signal sizes are between $(A_{0,\min}, A_{0,\max}) = (0.1, 0.4)$ (small signal setting) and $(A_{0,\min}, A_{0,\max}) = (0.4, 0.6)$ (large signal setting), respectively.
	Note that LL in the tables represents LANCE prior.
	In the small signal setting, YB method tends to outperform LANCE prior when $n=100$, while LANCE prior tends to outperform YB method when $n=300$.
	Although the performances of both of the methods are improved as $n$ increases, LANCE prior offers a more significant performance improvement.
	On the other hand, in the large signal setting, LANCE prior works better than YB method regardless of the sample size $n$.
	Compared with the small signal setting, LANCE prior shows slightly larger errors in the large signal setting.
	However, relatively, YB method produces much larger errors in the large signal setting.
	This is consistent with previous observations that under large signals, the performance of YB method is not satisfactory in terms of sensitivity and specificity.

	\begin{table}[!tb]
		\centering\footnotesize
		\caption{
			The mean and standard deviation (in parentheses) of the Frobenius norm, spectral norm and matrix $\ell_\infty$-norm when $(A_{0,\min}, A_{0,\max}) = (0.1, 0.4)$ based on 10 simulated data sets.
		}\vspace{.15cm}
		\begin{tabular}{c c c c c c c c}
			\hline 
			\multicolumn{2}{c}{Model 1}  & \multicolumn{3}{c}{$n=100$} & \multicolumn{3}{c}{$n=300$} \\ 
			& $p$ & $\| \cdot \|_F$ & $\| \cdot \|$ & $\| \cdot \|_\infty$ & $\| \cdot \|_F$ & $\| \cdot \|$ & $\| \cdot \|_\infty$ \\ \hline
			\multirow{2}{*}{LL} & $100$ & $2.737\, (0.264)$ & $1.002 \, (0.294)$  & $3.770 \,(2.368)$    & $1.204 \,(0.084)$ & $0.342 \,(0.047)$ & $0.631 \,(0.096)$  \\ 
			& $200$ & $4.468 \,(0.323)$ & $1.581 \,(0.247)$  & $8.058 \,(1.047)$    & $1.873 \,(0.093)$ & $0.454 \,(0.071)$ & $1.152 \,(0.409)$  \\ 
			\hline
			\multirow{2}{*}{YB} & $100$ & $2.201 \,(0.119)$ & $0.559 \,(0.045)$  & $1.051 \,(0.129)$    & $1.323 \,(0.048)$ & $0.347 \,(0.031)$  & $0.696 \,(0.103)$  \\ 
			& $200$ & $3.106 \,(0.114)$ & $0.583 \,(0.045)$  & $1.054 \,(0.096)$    & $1.923 \,(0.064)$ & $0.385 \,(0.034)$ & $0.758 \, (0.088)$  \\
			\hline 
			\multicolumn{2}{c}{Model 2}   & \multicolumn{3}{c}{$n=100$} & \multicolumn{3}{c}{$n=300$} \\ 
			& $p$ & $\| \cdot \|_F$ & $\| \cdot \|$ & $\| \cdot \|_\infty$ & $\| \cdot \|_F$ & $\| \cdot \|$ & $\| \cdot \|_\infty$ \\ \hline
			\multirow{2}{*}{LL} & $100$ & $2.550\, (0.330)$ & $0.997 \,(0.289)$  & $3.986 \,(2.103)$    & $1.097 \,(0.099)$ & $0.410 \,(0.054)$ & $1.135 \,(0.236)$  \\ 
			& $200$ & $5.771 \,(0.842)$ & $1.731 \,(0.280)$  & $7.803 \,(1.553)$    & $2.338 \,(0.090)$ & $0.602 \,(0.049)$ & $2.447 \,(0.450)$  \\ 
			\hline
			\multirow{2}{*}{YB} & $100$ & $2.403 \,(0.153)$ & $0.836 \,(0.104)$  & $2.718 \,(0.479)$    & $1.474 \,(0.129)$ & $0.549 \,(0.063)$  & $1.747 \,(0.331)$  \\ 
			& $200$ & $5.149 \,(0.176)$ & $1.418\, (0.107)$  & $6.597 \,(0.752)$    & $3.173 \,(0.194)$ & $0.925 \,(0.102)$ & $4.211 \,(0.649)$  \\
			\hline
			\multicolumn{2}{c}{Model 3}   & \multicolumn{3}{c}{$n=100$} & \multicolumn{3}{c}{$n=300$} \\ 
			& $p$ & $\| \cdot \|_F$ & $\| \cdot \|$ & $\| \cdot \|_\infty$ & $\| \cdot \|_F$ & $\| \cdot \|$ & $\| \cdot \|_\infty$ \\ \hline
			\multirow{2}{*}{LL} & $100$ & $3.884 \,(0.487)$ & $1.537 \,(0.289)$  & $6.987\, (1.521)$    & $1.641 \,(0.126)$ & $0.615 \,(0.057)$ & $2.469 \,(0.325)$  \\ 
			& $200$ & $7.784 \,(0.479)$ & $2.158\, (0.303)$  & $9.544 \,(1.203)$    & $3.328 \,(0.146)$ & $0.878 \,(0.096)$ & $3.795 \,(0.541)$  \\ 
			\hline
			\multirow{2}{*}{YB} & $100$ & $4.122 \,(0.396)$ & $1.528 \,(0.107)$  & $7.570\, (0.812)$    & $2.402\, (0.235)$ & $0.926 \,(0.125)$  & $4.360 \,(0.573)$  \\ 
			& $200$ & $8.195 \,(0.365)$ & $1.912 \,(0.078)$  & $8.424 \,(0.433)$    & $4.919 \,(0.323)$ & $1.273 \,(0.101)$ & $5.559 \,(0.662)$  \\
			\hline
		\end{tabular}\label{table:est1}
	\end{table}

	\begin{table}[!tb]
		\centering\footnotesize
		\caption{
			The mean and standard deviation (in parentheses) of the Frobenius norm, spectral norm and matrix $\ell_\infty$-norm when $(A_{0,\min}, A_{0,\max}) = (0.4, 0.6)$ based on 10 simulated data sets.
		}\vspace{.15cm}
		\begin{tabular}{c c c c c c c c}
			\hline 
			\multicolumn{2}{c}{Model 1}  & \multicolumn{3}{c}{$n=100$} & \multicolumn{3}{c}{$n=300$} \\ 
			& $p$ & $\| \cdot \|_F$ & $\| \cdot \|$ & $\| \cdot \|_\infty$ & $\| \cdot \|_F$ & $\| \cdot \|$ & $\| \cdot \|_\infty$ \\ \hline
			\multirow{2}{*}{LL} & $100$ & $2.749 \,(0.401)$ & $1.180\, (0.390)$  & $4.362\, (3.019)$    & $0.975\,( 0.057)$ & $0.310 \,(0.085)$ & $0.546 \,(0.161)$  \\ 
			& $200$ & $5.008 \,(0.379)$ & $1.922 \,(0.309)$  & $9.599 \,(1.486)$    & $1.918 \,(0.101)$ & $0.550 \,(0.084)$ & $1.457 \,(0.503)$  \\ 
			\hline
			\multirow{2}{*}{YB} & $100$ & $2.890 \,(0.196)$ & $0.853 \,(0.111)$  & $1.688 \,(0.258)$    & $1.590\, (0.092)$ & $0.483 \,(0.051)$  & $1.003 \,(0.144)$  \\ 
			& $200$ & $4.024 \,(0.202)$ & $0.898 \,(0.110)$  & $1.772 \,(0.226)$    & $2.284 \,(0.106)$ & $0.527 \,(0.054)$ & $1.091 \,(0.080)$  \\
			\hline
			\hline 
			\multicolumn{2}{c}{Model 2}   & \multicolumn{3}{c}{$n=100$} & \multicolumn{3}{c}{$n=300$} \\ 
			& $p$ & $\| \cdot \|_F$ & $\| \cdot \|$ & $\| \cdot \|_\infty$ & $\| \cdot \|_F$ & $\| \cdot \|$ & $\| \cdot \|_\infty$ \\ \hline
			\multirow{2}{*}{LL} & $100$ & $2.641 \,(0.472)$ & $1.135 \,(0.372)$  & $4.469 \,(2.390)$    & $1.021\, (0.116)$ & $0.410 \,(0.071)$ & $1.149 \,(0.286)$  \\ 
			& $200$ & $6.784 \,(1.156)$ & $2.141 \,(0.380)$  & $9.701 \,(1.953)$    & $2.599 \,(0.163)$ & $0.756 \,(0.115)$ & $3.195 \,(0.656)$  \\ 
			\hline
			\multirow{2}{*}{YB} & $100$ & $3.304 \,(0.345)$ & $1.313\, (0.215)$  & $4.461 \,(0.836)$    & $1.961 \,(0.251)$ & $0.850 \,(0.183)$  & $2.692 \,(0.726)$  \\ 
			& $200$ & $7.794\, (0.538)$ & $2.469\, (0.204)$  & $11.695  \,(1.185)$    & $4.319 \,(0.354)$ & $1.504 \,(0.226)$ & $6.654  \,(1.294)$  \\
			\hline
			\hline
			\multicolumn{2}{c}{Model 3}   & \multicolumn{3}{c}{$n=100$} & \multicolumn{3}{c}{$n=300$} \\ 
			& $p$ & $\| \cdot \|_F$ & $\| \cdot \|$ & $\| \cdot \|_\infty$ & $\| \cdot \|_F$ & $\| \cdot \|$ & $\| \cdot \|_\infty$ \\ \hline
			\multirow{2}{*}{LL} & $100$ & $4.435 \,(0.520)$ & $1.979\,( 0.330)$  & $8.841 \,(1.648)$    & $1.797 \,(0.226)$ & $0.778 \,(0.119)$ & $3.245 \,(0.479)$  \\ 
			& $200$ & $9.800\, (0.706)$ & $3.048 \,(0.520)$  & $13.509 \, (2.585)$    & $4.000 \,(0.248)$ & $1.228 \,(0.224)$ & $5.645 \,(1.594)$  \\ 
			\hline
			\multirow{2}{*}{YB} & $100$ & $6.245\,( 0.377)$ & $2.572 \,(0.121)$  & $12.646 \, (1.246)$    & $3.674 \,(0.529)$ & $1.652 \,(0.242)$  & $7.547 \,(1.154)$  \\ 
			& $200$ & $12.266  \,(0.796)$ & $3.169\, (0.248)$  & $13.943  \,(1.262)$    & $8.232 \,(0.727)$ & $2.386\,( 0.288)$ & $10.248  \,(1.570)$  \\
			\hline
		\end{tabular}\label{table:est2}
	\end{table}

	\bibliographystyle{dcu}
	\bibliography{refs}
	
\end{document}